\documentclass[a4paper,10pt]{article}

\usepackage{amsmath,amsgen,latexsym}
\usepackage{amstext,amssymb,amsfonts,latexsym}
\usepackage{theorem}
\usepackage{pifont}
\usepackage[pdftex]{graphics}
\usepackage{microtype}
\usepackage{graphicx}

 \newcommand{\ms}{\medskip}
 \newcommand{\n}{\noindent}
 \newcommand{\s}{\smallskip}
 \newcommand{\hs}[1]{\hspace*{ #1 mm}}
 \newcommand{\vs}[1]{\vspace*{ #1 mm}}

 \newcommand{\setempty}{\varnothing}
 
 \newcommand{\nat}{\mathbb{N}}
 
 \newcommand{\integer}{\mathbb{Z}}

 \newcommand{\co}{\mathrm{co}\mbox{-}}

 \newcommand{\AAA}{{\cal A}}

 \newcommand{\FF}{{\cal F}}
 \newcommand{\DD}{{\cal D}}

 \newcommand{\MM}{{\cal M}}

 \newcommand{\PP}{{\cal P}}

 \newcommand{\dl}{\mathrm{L}}
 \newcommand{\nl}{\mathrm{NL}}
 \newcommand{\p}{\mathrm{P}}
 \newcommand{\np}{\mathrm{NP}}

 \newcommand{\bpp}{\mathrm{BPP}}

 \newcommand{\fl}{\mathrm{FL}}

 \newcommand{\cfl}{\mathrm{CFL}}
 
 \newcommand{\dcfl}{\mathrm{DCFL}}

\theoremstyle{plain}
\theoremheaderfont{\bfseries}
 \newtheorem{theorem}{Theorem}[section]
 \newtheorem{lemma}[theorem]{Lemma}
 
 \newtheorem{corollary}[theorem]{Corollary}

 {\theorembodyfont{\rmfamily}
  }
 {\theorembodyfont{\rmfamily} }
 {\theorembodyfont{\rmfamily} }

 \newtheorem{claim}{Claim}

 \newenvironment{proof}{\par \noindent
            {\bf Proof. \hs{2}}}{\hfill$\Box$ \vspace*{3mm}}

 \newenvironment{proofof}[1]{\vspace*{5mm} \par \noindent
         {\bf Proof of #1.\hs{2}}}{\hfill$\Box$ \vspace*{3mm}}

 \newcommand{\ceilings}[1]{\lceil #1 \rceil}
 \newcommand{\floors}[1]{\lfloor #1 \rfloor}
 \newcommand{\pair}[1]{\langle #1 \rangle}

\newcommand{\ignore}[1]{}

\newcommand{\track}[2]{[\: \begin{subarray}{c} #1 \\%
      #2 \end{subarray} ]}

\newcommand{\ac}[1]{\mathrm{AC}^{ #1 }}

\newcommand{\logcfl}{\mathrm{LOGCFL}}
\newcommand{\logdcfl}{\mathrm{LOGDCFL}}

\newcommand{\stevec}{\mathrm{SC}}

\newcommand{\Lmreduces}{\leq^{\mathrm{L}}_{m}}

\begin{document}
\pagestyle{plain}
\setcounter{page}{1}

\begin{center}
{\Large {\bf Between SC and LOGDCFL: Families of Languages Accepted by Logarithmic-Space Deterministic Auxiliary Depth-$k$ Storage Automata}}\footnote{This exposition corrects and expands its preliminary report, which appeared in the Proceedings of the 27th International Conference on Computing and Combinatorics  (COCOON 2021), Tainan, Taiwan, October 24--26, 2021, Lecture Notes in Computer Science, Springer, vol.  13025, pp. 164--175, 2021. An oral presentation was given online due to the coronavirus pandemic.}
\ms\\

{\sc Tomoyuki Yamakami}\footnote{Affiliation: Faculty of Engineering, University of Fukui, 3-9-1 Bunkyo, Fukui 910-8507, Japan} \ms\\
\end{center}


\begin{abstract}
The closure of deterministic context-free languages under logarithmic-space many-one reductions ($\dl$-m-reductions), known as LOGDCFL, has been studied in depth from an aspect of parallel computability because it is nicely situated between $\dl$ and $\ac{1}\cap\mathrm{SC}^2$.
By replacing a memory device from pushdown stacks with access-controlled storage tapes, we introduce a computational model of one-way deterministic depth-$k$ storage automata ($k$-sda's) whose tape cells are freely modified during the first $k$ accesses and then become blank forever.
These $k$-sda's naturally induce the language family $k\mathrm{SDA}$. Similarly to $\mathrm{LOGDCFL}$, we study the closure $\mathrm{LOG}k\mathrm{SDA}$ of all languages in $k\mathrm{SDA}$ under $\dl$-m-reductions.
We demonstrate that $\dcfl\subseteq k\mathrm{SDA}\subseteq \mathrm{SC}^k$ by significantly extending Cook's early result (1979) of $\dcfl\subseteq \mathrm{SC}^2$.
The entire hierarch of $\mathrm{LOG}k\mathrm{SDA}$ for all $k\geq1$ therefore  lies between $\mathrm{LOGDCFL}$ and $\mathrm{SC}$.
As an immediate consequence, we obtain the same simulation bounds for Hibbard's limited automata.
We further characterize $\mathrm{LOG}k\mathrm{SDA}$ in terms of a new machine model, called  logarithmic-space deterministic auxiliary depth-$k$ storage automata that run  in polynomial time. These machines are  as powerful as a polynomial-time two-way multi-head deterministic depth-$k$ storage automata. We also provide a ``generic'' $\mathrm{LOG}k\mathrm{SDA}$-complete language under $\dl$-m-reductions by constructing a two-way universal simulator working for all $k$-sda's.

\s
\n{\bf Keywords.}
{parallel computation, deterministic context-free language, logarithmic-space many-one reduction, storage automata, LOGDCFL, SC, auxiliary  storage automata, multi-head storage automata, limited automata}
\end{abstract}

\sloppy

\section{DCFL, LOGDCFL, and Beyond}\label{sec:LOGDCFL-beyond}

In the literature, numerous computational models and associated language families have been proposed to capture various aspects of \emph{parallel computation}. Of those language families, we wish to pay special attention to the family known as $\mathrm{LOGDCFL}$, which is obtained from $\mathrm{DCFL}$, the family of all \emph{deterministic context-free (dcf) languages}, by taking the closure under logarithmic-space many-one reductions (or $\dl$-m-reductions, for short) \cite{Coo71,Sud78}.
These dcf languages were first defined in 1966 by Ginsburg and Greibach \cite{GG66} and their fundamental properties were studied intensively since then.
It is well known that $\dcfl$ is a proper subfamily of $\cfl$, the family of all context-free languages, because the context-free language $\{ww^R\mid w\in\{0,1\}^*\}$ (where $w^R$ means the reverse  of $w$), for instance,  does not belong to $\dcfl$. The dcf languages in general behave quite differently from the context-free languages. As an example of such differences, $\dcfl$ is closed under complementation, while $\cfl$ is not. This fact structurally distinguishes between $\dcfl$ and $\cfl$. Moreover, dcf languages require  computational resources of polynomial time and $O(\log^2{n})$ space simultaneously \cite{Coo79}; however, we do not know the same statement holds for context-free languages.
Although dcf languages are accepted by \emph{one-way deterministic pushdown automata} (or 1dpda's), these languages have a close connection to highly parallelized computation because of the nice inclusions
$\mathrm{NC}^1 \subseteq \dl \subseteq \nl \subseteq \logcfl \subseteq \mathrm{AC}^1$ and
$\dl\subseteq \logdcfl\subseteq \stevec^2$, and thus $\mathrm{LOGDCFL}$ has played a key role in discussing parallel complexity issues within $\p$.

It is known that $\mathrm{LOGDCFL}$ can be characterized  without using $\dl$-m-reductions by several other intriguing machine models, which include:
Cook's polynomial-time logarithmic-space \emph{deterministic auxiliary pushdown automata} \cite{Coo71},  \emph{two-way multi-head deterministic pushdown automata} running in polynomial time,
logarithmic-time CROW-PRAMs with polynomially many processors \cite{DR86},  and circuits made up of polynomially many multiplex select gates having logarithmic depth \cite{FLR96} or having polynomial proof-tree size \cite{MRV99}.
Such a variety of characterizations prove $\mathrm{LOGDCFL}$ to be a robust and fully-applicable notion in computer science.

Another important feature of $\logdcfl$ (as well as its underlying $\dcfl$) is the existence of ``complete'' languages, which are practically the most difficult languages in $\dcfl$ to recognize. Notice that a language $L$ is said to be \emph{$\dl$-m-complete} for a family $\FF$ of languages (or $L$ is \emph{$\FF$-complete}, for short) if $L$ belongs to $\FF$ and every language in $\FF$ is $\dl$-m-reducible to $L$.
Sudborough \cite{Sud78} first constructed such a language, called  $L_0^{(m)}$, which possesses the highest complexity (which he quoted as ``tape hardest'') among all dcf languages under $\dl$-m-reductions; therefore, $L_0^{(m)}$ is $\dl$-m-complete for $\dcfl$ and also for $\mathrm{LOGDCFL}$.
Using Sudborough's hardest languages, Lohrey \cite{Loh05} lately  presented another $\mathrm{LOGDCFL}$-complete problem based on semi-Thue systems.
Nonetheless, only a few languages are known today to be complete for $\dcfl$ as well as $\mathrm{LOGDCFL}$.

A large void seems to lie between $\dcfl$ and $\cfl$ (as well as $\co\cfl$). This void has been filled with, for example, the union hierarchy $\{\dcfl[k]\mid k\geq1\}$ and the intersection hierarchy $\{\dcfl(k)\mid k\geq1\}$ over $\dcfl$, where $\dcfl[k]$ (resp., $\dcfl(k)$) is composed of all unions (resp., intersections) of $k$ dcf languages. They truly form distinctive infinite hierarchies \cite{LW73,Yam20}.
Taking a quite different approach, Hibbard \cite{Hib67} devised specific rewriting systems, known as \emph{deterministic scan limited automata}. Those rewriting system were lately remodeled in \cite{PP14,PP15} as
single input/storage-tape 2-way deterministic linear-bounded automata that can modify the contents of their tape cells whenever the associated tape heads access the tape cells (when a tape head makes a turn, however, we count it twice); however, such  modifications are limited to only the first $k$ accesses and then the tape cells are \emph{frozen} forever.
Those machines are dubbed as \emph{deterministic $k$-limited  automata} (or $k$-lda's, for short).
Numerous followup studies, including Pighizzini and Prigioniero \cite{PP19}, Kutrib and Wendlandt \cite{KW17}, and Yamakami \cite{Yam19}, have lately revitalized an old study on $k$-lda's.
It is possible to impose on $k$-lda's the so-called \emph{blank-skipping property} \cite{Yam19}, by which inner states of $k$-sda's cannot be changed while reading any blank symbol.
A drawback of Hibbard's model is that the use of a single tape prohibits us from accessing input and memory simultaneously.
Seemingly, this drawback makes it difficult to construct a language that is ``hard'' for the family of all languages recognized by $k$-lda's in a way  similar to the construction of Sudborough's hardest languages.

It seems quite natural to seek out a reasonable extension of $\mathrm{LOGDCFL}$ by generalizing its underlying machines in a simple way.
A basis of $\mathrm{LOGDCFL}$ is of course 1dpda's, each of which is equipped with a read-once\footnote{A read-only tape is called \emph{read once} if,  whenever it reads a tape symbol (except for $\varepsilon$-moves, if any), it must move to the next unread cell.} input tape together with a storage device called a \emph{stack}.
Each stack allows two major operations. A pop operation is a deletion of a symbol and a push operation is a rewriting of a symbol on the topmost stack cell.
However, the usage of pushdown storage seems too restrictive in practice, and thus
various extensions of such pushdown automata have been sought in the past literature.
For instance, a \emph{stack automaton} of Ginsburg, Greibach, and Harrison \cite{GGH67a,GGH67b} is capable of freely traversing the inside of the stack to access each stored item but it is disallowed to  modify them unless the scanning stack head eventually comes to the top of the stack.
Thus, each cell of the stack could be accessed a number of times.
Meduna's \emph{deep pushdown automata} \cite{Med06} also allow stack heads to move deeper into the content of the stacks and to replace some symbols by appropriate strings. Other extensions of pushdown automata include \cite{Cha88,Iba71}.
To argue parallel computations, we intend to seek for a reasonable restriction of stack automata by replacing stacks with access-controlled storage devices. Each cell content of such a storage device is modified by its own tape head, which moves sequentially back and forth along the storage tape. This special tape and its tape head naturally allow  underlying machines to make more flexible memory manipulations.

In real-life circumstances, it seems reasonable to limit the number of times to access data sets stored in the storage device. For instance, rewriting data items into blocks of a memory device, such as external hard drives or rewritable DVDs, is usually costly and it may be restricted during each execution of a computer program.
We thus need to demand that every memory cell on this device can be modified only during the first few accesses and, in case of exceeding the intended access limit,
the storage cell turns unusable and no more rewriting is possible.
For simplicity, we refer to the number of times that the content of a storage cell is modified as ``depth''. We need to distinguish two circumstances depending on whether we allow or disallow a free access to ``new'' input symbols while scanning such unusable data sets. We later use the terms of
``immunity'' or ``susceptibility'' to the depth of storage devices.
We leave a further discussion on this issue to Section \ref{sec:TMs}.

To understand the role of depth limit for an underlying machines, let us consider how to recognize the non-context-free language $L_{abc}=\{a^nb^nc^{2n}\mid n\geq0\}$ under an additional requirement that new input symbols are only read while scanning storage cells are not yet frozen.
Given an input of the form $a^lb^mc^n$, we first write $a^l$ into the first $l$ cells of the storage device, check if $l=m$ by simultaneously reading $b^m$ and traversing the storage device  backward by changing $a$ to $b$, and then check if $l+m=n$ by simultaneously reading $c^n$ together with moving the device's scanning head back and forth by changing $b$ to $c$ and then $c$ to $B$ (frozen blank symbol). This procedure requires its depth limit to be $4$.

A storage device whose cells have depth at most $k$ is called a \emph{depth-$k$ storage tape} in this exposition and its scanning head is hereafter cited as a \emph{depth-$k$ storage-tape head} for convenience.
The machines equipped with those devices, where each storage-tape cell is initially ``empty'' and turned \emph{frozen blank} after exceeding its depth limit, are succinctly called   \emph{one-way deterministic depth-$k$ storage automata} (or $k$-sda's, for short).
Our $k$-sda's naturally expand Hibbard's $k$-lda's.\footnote{This claim comes from the fact that Hibbard's rewriting systems can be forced to satisfy the blank-skipping property without compromising their computational power \cite{Yam19}.}
The requirement of turning cells into frozen blank is imperative because, without it, the machines become as powerful as polynomial-time Turing machines. This statement directly follows from the fact that non-erasing stack automata can recognize the \emph{circuit value problem}, which is a $\p$-complete problem.

For convenience, we introduce the notation $k\mathrm{SDA}$ for each index $k\geq2$ to express the family of all languages recognized by those ``depth-susceptible'' $k$-sda's (for a more precise definition, see Section \ref{sec:TMs}) whereas the notation $k\mathrm{SDA}_{imm}$ is reserved for the language family induced by ``depth-immune'' $k$-sda's.
As the aforementioned example of $L_{abc}$ shows, $4\mathrm{SDA}$ contains even non-context-free languages.
Analogously to forming $\mathrm{LOGDCFL}$ from $\mathrm{DCFL}$, for any index $k\geq2$, we define  $\mathrm{LOG}k\mathrm{SDA}$ as the closure of $k\mathrm{SDA}$ under $\dl$-m-reductions.
It follows from the definitions that $\mathrm{LOGDCFL}\subseteq \mathrm{LOG}k\mathrm{SDA} \subseteq \mathrm{LOG}(k+1)\mathrm{SDA} \subseteq \p$.
Among many intriguing questions, we wish to raise the following three simple questions regarding our new language family $k\mathrm{SDA}$ as well as its $\dl$-m-closure $\mathrm{LOG}k\mathrm{SDA}$.

(1) What is the computational complexity of language families $k\mathrm{SDA}$ as well as $\mathrm{LOG}k\mathrm{SDA}$?

(2) Is there any natural machine model that can precisely characterize  $\mathrm{LOG}k\mathrm{SDA}$ in order to avoid the use of $\dl$-m-reductions?

(3) Is there any language that is $\dl$-m-complete for  $\mathrm{LOG}k\mathrm{SDA}$?

The sole purpose of this exposition is to partially answer these questions through Sections \ref{sec:auxiliary}--\ref{sec:upper-bounds} after a formal introduction of $k$-sda's in Section \ref{sec:basics}.

\section{Introduction of Storage Automata}\label{sec:basics}

We formally define a new computational model, dubbed as \emph{deterministic storage automata}, and present basic properties of them.

\subsection{Numbers, Sets, Languages, and Turing Machines}\label{sec:numbers-sets}

We begin with fundamental notions and notation necessary to introduce a new computation model of automata with depth-bounded storage tapes.

The two notations  $\integer$ and $\nat$ represent the set of all \emph{integers} and that of all \emph{natural numbers} (i.e., nonnegative integers), respectively. Given two numbers $m,n\in\integer$ with $m\leq n$, $[m,n]_{\integer}$ denotes the \emph{integer interval} $\{m,m+1,m+2,\ldots,n\}$. In particular, when $n\geq1$, we abbreviate $[1,n]_{\integer}$ as $[n]$.
We use the \emph{binary representations} of natural numbers. For such a representation $x$, the notation $(x)_{(2)}$ denotes the corresponding natural number of $x$. For instance, we obtain $(0)_{(2)}=0$, $(1)_{(2)}=1$, $(10)_{(2)}=2$, $(11)_{(2)}=3$, $(100)_{(2)}=4$, etc.
Given a set $S$, $\PP(S)$ denotes the \emph{power set} of $S$, namely, the set of all subsets of $S$.

An \emph{alphabet} is a finite nonempty set of ``symbols'' or ``letters.''
Given any alphabet $\Sigma$, a \emph{string} over $\Sigma$ is a finite sequence of symbols in $\Sigma$. The \emph{length} of a string $x$ is the total number of symbols in $x$ and it is denoted by $|x|$. The special notation $\varepsilon$ is used to express the \emph{empty string} of length $0$.
The notation $\Sigma^*$ denotes the set of all strings over $\Sigma$. A \emph{language} over $\Sigma$ is simply a subset of $\Sigma^*$.
As customarily, we freely identify a \emph{decision problem} with its corresponding language.
Given a string $x=x_1x_2\cdots x_n$ with $x_i\in\Sigma$ for all $i\in[n]$,  the \emph{reverse} of $x$ is $x_nx_{n-1}\cdots x_1$ and is denoted $x^R$.
For two strings $x$ and $y$ over the same alphabet, $x$ is said to be a \emph{prefix} of $y$ if there exists a string $z$ for which $y=xz$. In this case, $z$ is called a \emph{suffix} of $y$. Given a language $A$ over $\Sigma$, $Pref(A)$ denotes the set of all prefixes of any string in $A$, namely, $\{w\mid \exists y\in\Sigma^*[wy\in A]\}$.


We assume the reader's familiarity with multi-tape Turing machines and we abbreviate \emph{deterministic Turing machines} as DTMs. To handle a ``sub-linear'' space DTM, we need to assume that its input tape is read only and an additional rewritable \emph{index tape} is used to access a tape cell whose address is specified by the content of the index tape.
More precisely, a machine wants to access the $i$th input-tape cell, then it must write the binary expression of $i$ onto this index tape and then enters a designated ``query'' state, which triggers to retrieve the content of the $i$th input-tape cell and copy it into the last blank cell of the index tape so that the machine can freely read it.
Since the index tape requires $\ceilings{\log{|x|}}$ bits to specify each symbol of an input $x$, the space limitation of an underlying machine is thus applied to only work tapes. We assume the reader's familiarity with fundamental complexity classes, such as $\dcfl$, $\cfl$, $\dl$, $\nl$, and $\p$.
Another important complexity class is Steve's class $\mathrm{SC}^k$, for each index $k\geq1$, which is the family of all languages recognized by DTMs in polynomial time using $O(\log^k{n})$ space \cite{Coo71}. Let $\mathrm{SC}$ denote the union $\bigcup_{k\in\nat^{+}} \mathrm{SC}^k$. It follows that $\dl = \mathrm{SC}^1 \subseteq \mathrm{SC}^k\subseteq \mathrm{SC}^{k+1} \subseteq \mathrm{SC}\subseteq \p$.

In order to compute a function on strings, we further provide a DTM with an extra write-once output tape so that the machine produces output strings, where a tape is \emph{write-once} if its tape head never moves to the left and,  whenever its tape head writes a nonempty symbol, it must move to the right. All (total) functions computable by such DTMs in polynomial time using only logarithmic space form the function class, known as $\fl$.

Given two languages $L_1$ over alphabet $\Sigma_1$ and $L_2$ over $\Sigma_2$, we say that $L_1$ is \emph{$\dl$-m-reducible} to $L_2$ (denoted by $L_1\Lmreduces L_2$) if there exists a function $f$ computed by an appropriate polynomial-time DTM using only $O(\log{n})$ space such that, for any $x\in\Sigma_1^*$, $x\in L_1$ iff $f(x)\in L_2$. We say that $L_1$ is \emph{inter-reducible} to $L_2$ via $\dl$-m-reductions (denoted by $L_1\equiv^{\mathrm{L}}_{m} L_2$) if both $L_1\Lmreduces L_2$ and $L_2\Lmreduces L_1$ hold.


We write $\dcfl$ for the collection of all deterministic context-free (dcf) languages, which are recognized by one-way deterministic pushdown automata (or 1dpda's, for short). The notation $\logdcfl$ expresses the \emph{$\dl$-m-closure} of $\dcfl$, namely, $\{A\mid \exists B\in\dcfl[A\Lmreduces B]\}$.

In relation to our new machine model, introduced in Section \ref{sec:TMs}, we briefly explain Hibbard's  ``scan-limited automata'' \cite{Hib67}, which were lately reformulated by Pighizzini and Pisoni \cite{PP14,PP15} using a single-tape Turing machine model. This exposition follows their formulation. For any positive integer $k$, a \emph{deterministic $k$-limited automaton} (or a $k$-lda, for short) is a single-tape linear automaton with two endmarkers, where an input/work tape initially holds an input string and its  tape head can modify the content of each tape cell during the first $k$ visits (when the tape head makes a turn, we double count the visit) of the tape cell.

\subsection{Storage Tapes and Storage Automata}\label{sec:TMs}

We expand the standard model of pushdown automata by substituting its stack for a more flexible storage device, called a storage tape. Formally, a \emph{storage tape} is a semi-infinite rewritable tape whose cells are initially blank (filled with a symbol $\Box$) and are accessed sequentially by a tape head that can move back and forth along the tape by changing tape symbols as it passes through.

In what follows, we fix a constant $k\in\nat^{+}$. A \emph{one-way  deterministic depth-$k$ storage automaton} (or a $k$-sda, for short) $M$ is a 2-tape DTM (equipped only with a read-only input tape and a rewritable work tape) of the form $(Q,\Sigma, \{\Gamma^{(e)}\}_{e\in[0,k]_{\integer}}, \{{\triangleright,\triangleleft}\}, \delta,q_0,Q_{acc},Q_{rej})$ with a finite set $Q$ of inner states, an input alphabet $\Sigma$, storage alphabets $\Gamma^{(e)}$ for indices  $e\in[0,k]_{\integer}$ with $\Gamma = \bigcup_{e\in[0,k]_{\integer}} \Gamma^{(e)}$, a transition function $\delta$ from $Q\times \check{\Sigma}\times \Gamma$ to $Q\times \Gamma \times D_1\times D_2$ with $\check{\Sigma} =\Sigma\cup\{{\triangleright,\triangleleft}\}$,  $D_1=\{0,+1\}$,  and $D_2=\{-1,0,+1\}$, an initial state $q_0$ in $Q$, and sets $Q_{acc}$ and $Q_{rej}$ of accepting states and rejecting states, respectively, satisfying both $Q_{acc}\cup Q_{rej}\subseteq Q$ and $Q_{acc}\cap Q_{rej}=\setempty$,
provided that $\Gamma^{(0)} =\{\Box\}$ (where $\Box$ is a distinguished initial blank symbol), $\Gamma^{(k)}=\{{\triangleright},B\}$ (where $B$ is a unique frozen blank symbol) and $\Gamma^{(e_1)}\cap \Gamma^{(e_2)}=\setempty$ for any distinct pair $e_1,e_2\in[0,k]_{\integer}$. The sets $D_1$ and $D_2$ indicate the directions of the input-tape head and those of the storage-tape head, respectively.
The choice of $D_1$ forces the input tape to be read once. We say that the input tape is \emph{read once} if its tape head either moves to the right or stays still with scanning no input symbol.
A single move (or step) of $M$ is dictated by $\delta$. If $M$ is in inner state $q$, scanning $\sigma$ on the input tape and $\tau$ on the storage tape, a transition $\delta(q,\sigma,\gamma) = (p,\xi,d_1,d_2)$ forces $M$ to change $q$ to $p$, write $\xi$ over $\gamma$, and move the input-tape head in direction $d_1$ and the storage-tape head in direction $d_2$.
The \emph{depth value} $dv(\gamma)$ of a symbol $\gamma$ is the index $e$ for which $\gamma\in\Gamma^{(e)}$.

All tape cells are indexed by natural numbers from left to right, where the leftmost tape cell is the \emph{start cell}, which is indexed $0$.
An input tape has endmarkers $\{{\triangleright,\triangleleft}\}$ and a storage tape has only  the left endmarker $\triangleright$.
When an input string $x$ is given to the input tape, it should be surrounded by the two endmarkers as $\triangleright\, x\, \triangleleft$ so that $\triangleright$ is located at the start cell and $\triangleleft$ is at the cell indexed $|x|+1$.
For any index $i\in\nat$, $x_{(i)}$ denotes the tape symbol written on the $i$th input-tape cell, provided that $x_{(0)} ={\triangleright}$ (left endmarker) and $x_{(n+1)}={\triangleleft}$ (right endmarker).
Similarly, when $z$ represents the non-$\Box$ portion of the content of a storage tape, the notation $z_{(i)}$ expresses the symbol written in the $i$th tape cell. In particular, $z_{(0)}={\triangleright}$.

For the storage tape, we apply the following rewriting rules.
Whenever the storage-tape head passes through a tape cell containing a symbol in $\Gamma^{(e)}$ with $e<k$, the machine must replace it by another symbol in $\Gamma^{(e+1)}$ except for the case of the following ``turns''.
Remember that no symbol in $\Gamma^{(k)}$ can be modified at any time.
We distinguish two types of turns.
Although we use various storage alphabets $\Gamma^{(0)},\Gamma^{(1)},\ldots,\Gamma^{(k)}$, since those alphabets are mutually disjoint, we can easily discern from which direction the tape head arrives simply by reading a storage tape symbol written in each tape cell.
As customary, we explicitly demand that, on a storage tape, no machine writes $\triangleright$ over any non-$\triangleright$ symbol.
A \emph{left turn at step $t$} refers to $M$'s step at which, after $M$'s tape head moves to the right at step $t-1$,
it moves to the left at step $t$.
In other words, $M$ takes two transitions $\delta(q,\sigma_1,\gamma_1)=(p,\tau,d,+1)$ at step $t-1$ and $\delta(p,\sigma_2,\gamma_2)=(r,\xi,e,-1)$ at step $t$.
Similarly, we say that $M$ makes a \emph{right turn at step $t$} if $M$'s tape head moves from the left at step $t-1$ and changes its direction to the right at step $t$. Whenever a tape head makes a turn, we treat such a case as ``double accesses.''
More formally, at a turn, any symbol in $\Gamma^{(e)}$ with $e<k$ must be changed to another symbol in $\Gamma^{(\min\{k,e+2\})}$.

The \emph{depth-$k$ requirement} demands that, assuming that $M$ modifies $\sigma\in\Gamma^{(e)}$ written on a storage tape with $e\in[0,k-1]_{\integer}$ to $\tau$ and moves its storage-tape head in direction $d$, if $d= (-1)^e$ (i.e., making no turn), then $\tau$ must belong to $\Gamma^{(\min\{e+1,k\})}$, and, if $d = (-1)^{e+1}$ (i.e., making a turn), then $\tau$ must be in $\Gamma^{(\min\{e+2,k\})}$.
A storage tape that satisfies the depth-$k$ requirement is succinctly called a \emph{depth-$k$ storage tape}.

Instead of making $\varepsilon$-moves (i.e., a tape head neither moves nor reads any tape symbol), we allow a tape head to make a \emph{stationary move},\footnote{The use of stationary move is made in this exposition only for convenience sake. It is also possible to define $k$-sda's using $\varepsilon$-moves in place of stationary moves.}
by which the tape head stays still and the currently scanned symbol is unaltered.
The tape head direction ``$0$'' indicates such a stationary move.
For clarity, if the input-tape (resp., the storage-tape) head makes a stationary move, then we call this move an \emph{input-stationary move} (resp., a \emph{storage-stationary move}).

Let us consider two different models whose input-tape head is either   ``depth-susceptible'' or ``depth-immune'' to the content of each storage-tape cell. Recall the \emph{blank-skipping property} of a $k$-lda \cite{Yam20}, which in essence states that, while the $k$-lda is scanning the blank symbol, it should not change its inner state.
The $k$-sda $M$ is called \emph{depth-susceptible} if, for any storage symbol  $\gamma$ scanned currently, (i) if $\gamma$ is in $\Gamma^{(k-1)}\cup\Gamma^{(k)}$, then the input-tape head must make a stationary move, and (ii) if $\gamma$ is frozen blank, then the current inner state should not change; namely, for any transition $\delta(q,\sigma,\gamma) = (p,\xi,d_1,d_2)$, (i$'$) $\gamma\in\Gamma^{(k-1)}\cup \Gamma^{(k)}$ implies $d_1=0$ and (ii$'$) $\gamma=B$ implies $q=p$.
On the contrary, the machine is \emph{depth-immune} if there is no restriction.

A \emph{surface  configuration} of $M$ on input $x$ is of the form $(q,l_1,l_2,z)$ with $q\in Q$, $l_1\in[0,|x|+1]_{\integer}$, $l_2\in\nat$, and $z\in(\Gamma-\{\Box\})^*$, which indicates the situation where $M$ is in inner state $q$, the storage tape contains $z$ (except for the tape symbol $\Box$), and two tape heads scan the $l_1$th cell of the input tape and the $l_2$th cell of the storage tape.

The \emph{initial surface configuration} has the form $(q_0,\triangleright,0,0)$ and $\delta$ describes how to reach the next surface configuration
in a single step. For convenience, we define the \emph{depth value} $dv(C)$ of a surface configuration $C=(q,l_1,l_2,z)$ to be the depth value $dv(z_{(l_2)})$.
An \emph{accepting surface configuration} (resp., a \emph{rejecting surface configuration}) is of the form $(q,h_1,h_2,w)$ with $q\in Q_{acc}$ (resp., $q\in Q_{rej}$). A \emph{halting configuration} means either an accepting configuration or a rejecting configuration.
A \emph{computation} of $M$ on input $x$ starts with the initial surface  configuration with the input $x$ and either ends with a halting surface configuration or continues forever.
The $k$-sda $M$ \emph{accepts} (resp., \emph{rejects}) $x$ if the computation of $M$ on $x$ reaches an accepting surface configuration (resp., a rejecting surface configuration).
For readability, we will drop the word ``surface'' altogether in the subsequent sections since we use only surface configurations.

For a language $L$ over $\Sigma$, we say that $M$ \emph{recognizes} (accepts or solves) $L$ if, for any input string $x\in\Sigma^*$, (i) if $x\in L$, then $M$ accepts $x$ and (ii) if $x\notin L$, then $M$ rejects $x$.
This implies that $M$ halts within finite steps on all inputs.
For two $k$-sda's $M_1$ and $M_2$ over the same input alphabet $\Sigma$, we say that $M_1$ is \emph{(computationally) equivalent} to $M_2$ if, for any input $x\in\Sigma^*$, $M_1$ accepts (resp., rejects) $x$ iff $M_2$ accepts (resp., rejects) $x$.

For notational convenience, we write $k\mathrm{SDA}$ for the collection of all languages recognized by depth-susceptible  $k$-sda's and $\mathrm{LOG}k\mathrm{SDA}$ for the collection of languages that are $\dl$-m-reducible to certain languages in $k\mathrm{SDA}$.
Moreover, we set  $\omega\mathrm{SDA}$ to be the union $\bigcup_{k\in\nat^{+}} k\mathrm{SDA}$.
With this notation, the non-context-free language $L_{abc}$, discussed in Section \ref{sec:LOGDCFL-beyond}, belongs to $4\mathrm{SDA}$. Thus, $4\mathrm{SDA}\nsubseteq \cfl$ follows instantly.
Based on the depth-immune model of $k$-sda, we similarly define $k\mathrm{SDA}_{imm}$, $\omega\mathrm{SDA}_{imm}$, and  $\mathrm{LOG}k\mathrm{SDA}_{imm}$.


The depth-susceptibility plays a crucial role in showing the following equality between $2\mathrm{SDA}$ and $\mathrm{DCFL}$. This fact supports the claim that $k$-sda's naturally expand 1dpda's. It follows that $\dcfl \subseteq k\mathrm{SDA} \subseteq (k+1)\mathrm{SDA} \subseteq \omega\mathrm{SDA}$ for any $k\geq2$.

\begin{lemma}\label{2SDA=DCFL}
$2\mathrm{SDA} = \mathrm{DCFL}$.
\end{lemma}

\begin{proof}
It was shown that $\dcfl$ is precisely characterized by deterministic 2-limited automata (or 2-lda's, for short) \cite{Hib67,PP15}. Since any 2-lda can be transformed to another 2-lda with the blank-skipping property \cite{Yam19}, depth-susceptible 2-sda's can simulate 2-lda's; thus, we immediately conclude that $\dcfl\subseteq 2\mathrm{SDA}$.

For the converse, it suffices to simulate depth-susceptible 2-sda's by appropriate 1dpda's.
Given a depth-susceptible 2-sda $M$, we design a one-way deterministic pushdown automaton (or a 1dpda) $N$ that works as follows. We want to treat the storage-tape of $M$ as a stack by ignoring the frozen blank symbol $B$ in $\Gamma^{(2)}$.
When $M$ modifies the initial blank symbol $\Box$ on its storage tape to a new symbol $\sigma$ in $\Gamma^{(1)}$, $N$ pushes $\sigma$ to a stack. In the case where $M$ modifies a storage symbol $\sigma$ to $B$, $N$ pops the same symbol $\sigma$ from the stack.
Note that, since $M$ is depth-susceptible, it cannot move the input-tape head. While $M$ scans $B$, $N$ does nothing.
As for the behavior of $N$'s input-tape head, if $M$ reads an input symbol $\sigma$ and moves to the right, then $N$ does the same.
On a tape cell containing $\sigma\in \Gamma^{(1)}$, since $M$'s input-tape head makes a series of stationary moves, $N$ reads $\sigma$ at the first move, remembers $\sigma$, and makes $\varepsilon$-moves afterwards until $M$ ends its stationary moves.   Obviously, the resulting machine is a 1dpda and it precisely simulates $M$.
\end{proof}

\n{\bf Remark.} Remember that storage-tape cells of $k$-sda's become blank after the first $k$ accesses. If we allow such tape cells to freeze the last written symbols and preserve them forever, instead of erasing them, then the resulting machines get enough power to recognize even the circuit value problem, which is $\p$-complete under $\dl$-m-reductions.
In this exposition, we do not further delve into this topic.


\section{Two Machine Models that Characterize LOG$k$SDA}\label{sec:auxiliary}

We begin with a  study on the structural properties of $\mathrm{LOG}k\mathrm{SDA}$, which is the closure of $k\mathrm{SDA}$ under $\dl$-m-reductions. In particular, we intend to seek out different characterizations of $\mathrm{LOG}k\mathrm{SDA}$ \emph{with no use of $\dl$-m-reductions}. The basic idea of the elimination of such reductions is attributed to Sudborough \cite{Sud78}, who characterized $\logdcfl$ using two machine models: polynomial-time log-space auxiliary deterministic pushdown automata and polynomial-time multi-head deterministic pushdown automata.
Our goal in this section is to expand these machine models to fit into the framework of depth-$k$ storage automata and prove their characterizations of $\mathrm{LOG}k\mathrm{SDA}$.

\subsection{Deterministic Auxiliary Depth-$k$ Storage Automata}\label{sec:aux-k-sda}

We expand deterministic auxiliary pushdown automata to \emph{deterministic auxiliary depth-$k$ storage automata}, each of which is equipped with a two-way read-only input tape, a rewritable auxiliary (work) tape, and a depth-$k$ storage tape.

Let us formally formulate \emph{deterministic auxiliary depth-$k$ storage automata} (or an aux-$k$-sda, for short).
For the description of such a machine, firstly we prepare a two-way read-only input tape and a depth-$k$ storage tape and secondly we supply a new space-bounded rewritable auxiliary tape whose cells are freely modified by a two-way tape head.
An aux-$k$-sda $M$ is a 3-tape DTM $(Q,\Sigma,\Theta, \{\Gamma^{(e)}\}_{e\in[0,k]_{\integer}}, \{\triangleright,\triangleleft\}, \delta,q_0,Q_{acc},Q_{rej})$ with a read-only input tape, an auxiliary  rewritable (work) tape with an alphabet $\Theta$, and a depth-$k$ storage tape.
Initially, the input tape is filled with ${\triangleright\, x \, \triangleleft}$, the auxiliary  tape is blank, and the depth-$k$ storage tape has only the  initial blank symbols $\Box$ except for the left endmarker $\triangleright$.
We set $\Gamma^{(0)}= \{\Box\}$, $\Gamma^{(k)}=\{\triangleright,B\}$, and $\Gamma = \bigcup_{e\in[0,k]_{\integer}}\Gamma^{(e)}$, provided that $\Gamma^{(e_1)}\cap \Gamma^{(e_2)}=\setempty$ for any distinct pair $e_1,e_2\in[0,k]_{\integer}$.
The transition function $\delta$ of $M$ maps $(Q-Q_{halt}) \times \check{\Sigma} \times \Theta \times \Gamma$ to $Q\times \Theta \times \Gamma \times D_1\times D_2\times D_3$, where $Q_{halt} = Q_{acc}\cup Q_{rej}$, $\check{\Sigma} = \Sigma\cup\{\triangleright,\triangleleft\}$, and $D_1=D_2=D_3=\{-1,0,+1\}$.
A transition $\delta(q,\sigma,\tau,\gamma) = (p,\theta,\xi,d_1,d_2,d_3)$ indicates that, on reading input symbol $\sigma$, $M$ changes its inner state $q$ to $p$ by moving an input-tape head in direction $d_1$, changes auxiliary tape symbol $\tau$ to $\theta$ by moving an auxiliary-tape head in direction $d_2$, and changes storage tape symbol $\gamma$ to $\xi$ by moving in direction $d_3$. A string $x$ is \emph{accepted} (resp., \emph{rejected}) if $M$ enters an inner state in $Q_{acc}$ (resp., $Q_{rej}$).

When excluding $(\Theta,D_2)$ from the definition of $M$, the resulting automaton must fulfill the \emph{depth-$k$ requirement} of $k$-sda's given in Section \ref{sec:TMs}. The \emph{depth-susceptibility condition} for an aux-$k$-sda is stated as follows:
For a transition $\delta(q,\sigma,\tau,\gamma) = (p,\theta,\xi,d_1,d_2,d_3)$, (i) if $\gamma\in\Gamma^{(k-1)}\cup\Gamma^{(k)}$, then $d_1=d_2=0$ and (ii) if $\gamma=B$, then $q=p$.
Regarding stationary moves, $M$ should satisfy the \emph{stationary requirement} that, assuming that $\delta(q,\sigma,\tau,\gamma) = (p,\theta,\xi,d_1,d_2,d_3)$, $d_2=0$ implies $\tau=\theta$ and $d_3=0$ implies $\gamma=\xi$.

Let us assume that, for a certain positive integer $c$, $M$'s auxiliary tape uses at most $c\ceilings{\log{n}}$ tape cells on any input of length $n\geq2$. It is then possible to reduce this space bound down to $\floors{\log{n}}$. We first introduce a larger auxiliary tape alphabet $\Theta'$ so that the content of $c$ consecutive tape cells can be expressed as a single symbol in $\Theta'$.
Since a tape head cannot access all $c$ tape cells at once,
we need to specify which tape cell is currently scanned by the tape head.
For this purpose, we use a number $i\in [c]$ to refer to the $i$th tape cell and make an inner state hold the information on $i$. Therefore, without loss of generality, we can assume that $M$ uses at most $\floors{\log{n}}$ tape cells on the auxiliary tape.

\subsection{Multi-Head Deterministic Depth-$k$ Storage Automata}\label{sec:multi-head}

We further introduce another useful machine model by expanding two-way multi-head pushdown automata to \emph{two-way multi-head deterministic depth-$k$ storage  automata}, each of which allows two-way multiple tape heads to read a given input.

For each fixed integer $\ell\geq1$, we define an \emph{$\ell$-head deterministic depth-$k$ storage automaton} as a 2-tape DTM with $\ell$ read-only depth-susceptible tape heads scanning over an input tape and another tape head over a depth-$k$ storage tape.
For convenience, we call such a machine a \emph{$k$-sda$_{2}$($\ell$)}, where the subscript ``$2$'' emphasizes that all tape heads move in both directions (including stationary moves). Remember that each $k$-sda$_{2}$($\ell$) has actually $\ell+1$ tape heads, including one tape head working along the storage tape.

More formally, a $k$-sda$_{2}$($\ell$) is a tuple $(Q,\Sigma, \{\Gamma^{(e)}\}_{e\in[0,k]_{\integer}}, \{\triangleright, \triangleleft\}, \delta,q_0,Q_{acc},Q_{rej})$ with a transition function $\delta$ mapping $Q\times \check{\Sigma}^{\ell}\times \Gamma$ to $(Q-Q_{halt}) \times \Gamma \times D^{\ell}\times D$, where $Q_{halt} = Q_{acc}\cup Q_{rej}$, $\check{\Sigma} = \Sigma\cup\{{\triangleright,\triangleleft}\}$, $D=\{-1,0,+1\}$, $\Gamma^{(0)}=\{\Box\}$, $\Gamma^{(k)}=\{\triangleright,B\}$, and  $\Gamma=\bigcup_{e\in[0,k]_{\integer}} \Gamma^{(e)}$,
provided that $\Gamma^{(e_1)}\cap \Gamma^{(e_2)}=\setempty$ for any distinct pair $e_1,e_2\in[0,k]_{\integer}$.
A transition of the form  $\delta(q,\sigma_1,\ldots,\sigma_{\ell},\gamma) = (p,\xi,d_1,\ldots,d_{\ell},d_{\ell+1})$ means that, if  $M$ is in inner state $q$, scanning a tuple $(\sigma_1,\ldots,\sigma_{\ell})$ of symbols on the input tape by the $\ell$ read-only tape heads as well as symbol $\gamma$ on the depth-$k$ storage tape by the rewritable tape head, then, in a single step, $M$ enters inner state $p$ and writes $\xi$ over $\gamma$ by moving the $i$th input-tape head in direction $d_i$ for every index $i\in[\ell]$ and the storage-tape head in direction $d_{\ell+1}$.
Note that a $k$-sda$_{2}$($1$) and a $k$-sda are similar in their machine structures but the former can move its input-tape head to the left.

The acceptance/rejection criteria are the same as those of underlying $k$-sda's. We demand that all $k$-sda$_{2}$($\ell$) should satisfy the  \emph{depth-susceptibility condition}, which asserts that, for any transition $\delta(q,\sigma_1,\sigma_2,\ldots,\sigma_{\ell},\gamma) =(p,\xi,d_1,d_2,\ldots,d_{\ell},d_{\ell+1})$ of $M$, (i) $\gamma\in \Gamma^{(k-1)}\cup \Gamma^{(k)}$ implies $d_1=d_2=\cdots =d_{\ell}=0$ and (ii) $\gamma=B$ implies $q=p$.

\subsection{Characterization Theorem}

We intend to demonstrate that the two new machine models introduced in Sections \ref{sec:aux-k-sda}--\ref{sec:multi-head} precisely characterize $\mathrm{LOG}k\mathrm{SDA}$. This result can be seen as a natural extension of Sudborough's machine characterization of $\mathrm{LOGDCFL}$ to $\mathrm{LOG}k\mathrm{SDA}$.

\begin{theorem}\label{auxiliary-characterize}
Let $k\geq2$. Let $L$ be any language. The following three statements are logically equivalent.
\renewcommand{\labelitemi}{$\circ$}
\begin{enumerate}\vs{-2}
  \setlength{\topsep}{-2mm}%
  \setlength{\itemsep}{0mm}%
  \setlength{\parskip}{0cm}%

\item $L$ is in $\mathrm{LOG}k\mathrm{SDA}$.

\item There exists an aux-$k$-sda that recognizes $L$ in polynomial time using logarithmic space.

\item There exist a number $\ell\geq2$ and a  $k$-sda$_{2}$($\ell$)  that recognizes $L$ in polynomial time.
 \end{enumerate}
\end{theorem}

In the rest of this subsection, we intend to prove Theorem \ref{auxiliary-characterize}.
Firstly, Sudborough's proof \cite[Lemmas 3--6]{Sud78} for $\mathrm{LOGDCFL}$ relies on the heavy use of stack operations, which are applied only to the topmost symbol of the stack but the other symbols in the stack are intact.
In our case, however, we need to deal with the operations of a storage-tape head, which can move back and forth along a storage tape by modifying cell's content as many as $k$ times.
Secondly, Sudborough's characterization utilizes a simulation procedure  \cite[pp.338--339]{Har72} of Hartmanis and a proof argument \cite[Lemma 4.3]{Gal77} of Galil; however, we cannot directly use them, and thus a new idea is needed to establish Theorem \ref{auxiliary-characterize}.
The proof of this theorem therefore requires technically challenging simulations among $\fl$-functions and the other machine models of  aux-$k$-sda and $k$-sda$_2(\ell)$.


\begin{lemma}\label{two-input-model}
Given a function $f:\Sigma_1^*\to\Sigma_2^*$ in $\fl$ for certain alphabets $\Sigma_1$ and $\Sigma_2$ and a depth-susceptible $k$-sda $M$ working over $\Sigma_2$ in polynomial time,  there exists a log-space aux-$k$-sda $N$ that recognizes $L=\{x\in\Sigma_1^*\mid f(x)\in L(M) \}$ in polynomial time.
\end{lemma}

\begin{proof}
Let $f:\Sigma_1^*\to\Sigma_2^*$ be any function in $\fl$.
We take a DTM $M_f$, equipped with an read-only input tape, a logarithmic space-bounded rewritable work tape, and a write-once output tape, and assume that $M_f$ computes $f$ in polynomial time. A given depth-susceptible $k$-sda running in polynomial time is denoted $M$. We then set $L$ to be the language $\{x\mid f(x)\in L(M)\}$.

We design the desired aux-$k$-sda $N$ for $L$ as follows.
Given any input $x$ in $\Sigma_1^*$, we repeat the following process until $M$ enters a certain halting state. Using an auxiliary tape of $N$, we keep track of the content on $M_f$'s work tape, two head positions of $M$'s input and storage tapes and $M_f$'s input tape.
Since $N$ shares the same input $x$ with $M$,
$N$ can precisely simulate the movement of $M$'s input-tape head.
We intend to force $N$ to scan the same storage symbol as $M$. Let $\gamma$ denote the storage symbol scanned currently by $M$. Assume that $M$'s input-tape head is located at cell $h\in[0,|f(x)|+1]_{\integer}$.

\s

(1) If $M$ makes an input-stationary move, then we simply simulate one step of the behavior of $M$'s storage-tape head since we can reuse the last produced output symbol of $M_f$. We thus move the storage-tape head of $N$ in the same direction as $M$.

(2) Assume that $M$ moves its input-tape head from the left to cell $h$. We remember the current location of $M$'s input-tape head, return $N$'s input-tape head to the last location of $M_f$'s input-tape head, and resume the simulation of $M_f$ with the use of $N$'s auxiliary tape as a work tape by making a (possible) series of storage-stationary moves of $N$ until $M_f$ produces the $h$th output symbol, say, $\sigma$. This is possible because the  output tape of $M_f$ is write-once, and thus there is no need to recompute any past output symbols.
Once $\sigma$ is obtained, $N$ remembers the positions of $M_f$'s tape heads, moves its input-tape head back to the previous location. We then simulate a single step of $M$ on $\sigma$ together with  the storage symbol.
We then update the tape-head positions.

\s

It is not difficult to show that $N$ eventually reaches the same type (accepting or rejecting) of halting states as $M$ does within polynomially many steps. Notice that the storage-tape head and the auxiliary-tape head do not work simultaneously.
The depth-susceptibility of $N$ comes from that of $M$ since $M_f$ is simulated only when $M$'s storage-tape head reads a symbol not in $\Gamma^{(k-1)}\cup\Gamma^{(k)}$. Thus, $N$ is indeed an aux-$k$-sda.
\end{proof}


Given an aux-$k$-sda $M$, we intend to construct a $k$-sda$_2(\ell)$, which mimics the behavior of a given  aux-$k$-sda, where $\ell$ is an appropriate constant depending only on $M$.

\begin{lemma}\label{aux-to-6-heads}
Let $k\geq2$. Let $M$ denote a polynomial-time log-space aux-$k$-sda, there are a constant $c>0$ and a $k$-sda$_{2}(5c+2)$ $N$ that simulates $M$ in polynomial time.
\end{lemma}

\begin{proof}
Let $k\geq2$ and let $M = (Q,\Sigma, \{\Gamma^{(e)}\}_{e\in[0,k]_{\integer}}, \{\triangleright, \triangleleft\}, \delta,q_0,Q_{acc},Q_{rej})$ denote any aux-$k$-sda that runs in polynomial time using logarithmic space on all inputs of length $n\geq2$.

As noted in Section \ref{sec:aux-k-sda}, we can assume that $M$ uses at most $\floors{\log{n}}$ auxiliary-tape cells, where $n$ indicates input length. Let $c=\ceilings{\log|\Theta|}$ and assume that $\Theta=\{\theta_1,\theta_2,\ldots,\theta_{2^c}\}$. We identify each element $\theta_i$ in $\Theta$ with $i$, which is further expressed as its binary string of length $c$. We partition every auxiliary tape cell into $c$ blocks, indexed by $1,2,\ldots,c$. A symbol $\theta_i$ can be stored in those blocks in such a way that, for any $j\in[c]$, the $j$th bit of $\theta_i$ is placed in the $j$th block. By fixing each block index for all tape cells, we can partition an entire tape into $c$ separate ``tapes'', which are customarily called \emph{tracks}.

We want to construct a polynomial-time $k$-sda$_{2}(5c+2)$ $N$ for which $L(N)$ coincides with $L(M)$.
Other than a storage-tape head, we use $M$'s input-tape head as the \emph{principal} tape head of $N$. We further introduce additional $5c$ tape heads to simulate the behavior of an auxiliary-tape head of $M$.

In what follows, we fix one of the $c$ tracks of the auxiliary tape. If each track contains a string $w$ of length $\ceilings{\log{n}}$, we treat it as the binary number $(1w^R)_{(2)}$.
We use two tape heads to remember the positions of the input-tape head and the auxiliary-tape head of $M$. To remember the number $(1w^R)_{(2)}$, we need  additional $5$ tape heads (other than the storage-tape head).

Head 1 keeps the tape head position. Whenever $M$ moves the auxiliary-tape head, $N$ moves head 1 as well. Head 2 moves backward to measure the distance $l$ of the auxiliary-tape head from the left end of the auxiliary tape. Using this information, $N$ moves head 3 as follows. If the tape head changes $0$ to $1$ (resp., $1$ to $0$) on this target track, then $N$ moves the head $2^l$ cells to the right (resp., to the left).
How can we move a tape head to cell $2^{l}$ from $\triangleright\:$? Following  \cite{Har72}, we use three tape heads to achieve this goal in the following way.
We move head 4 one cell to the right. As head 4 takes one step on the way back to $\triangleright$, we move head 5 two cells to the right. We then switch the roles of heads 4 and 5. As head 5 takes one step back toward  $\triangleright$, we move head 4 two cells to the right. If we repeat this process $t$ times, one of the heads indeed reaches cell $2^t$. Hence, for the $l$th run, head 3 reaches cell $2^l$. This process requires 3 tapes. Thus, the total of $5$ tape heads are sufficient to simulate the operation on the track content of the auxiliary tape.

Since there are $c$ tracks, we need the total of $5c+2$ heads (including the input-tape head and the storage-tape head) for our intended simulation. Note that, while the additional $5c$ tape heads are moving, the storage-tape head scans no symbol in $\Gamma^{(k-1)}\cup \Gamma^{(k)}$.
\end{proof}


\begin{figure}[t]
\centering
\includegraphics*[height=4.5cm]{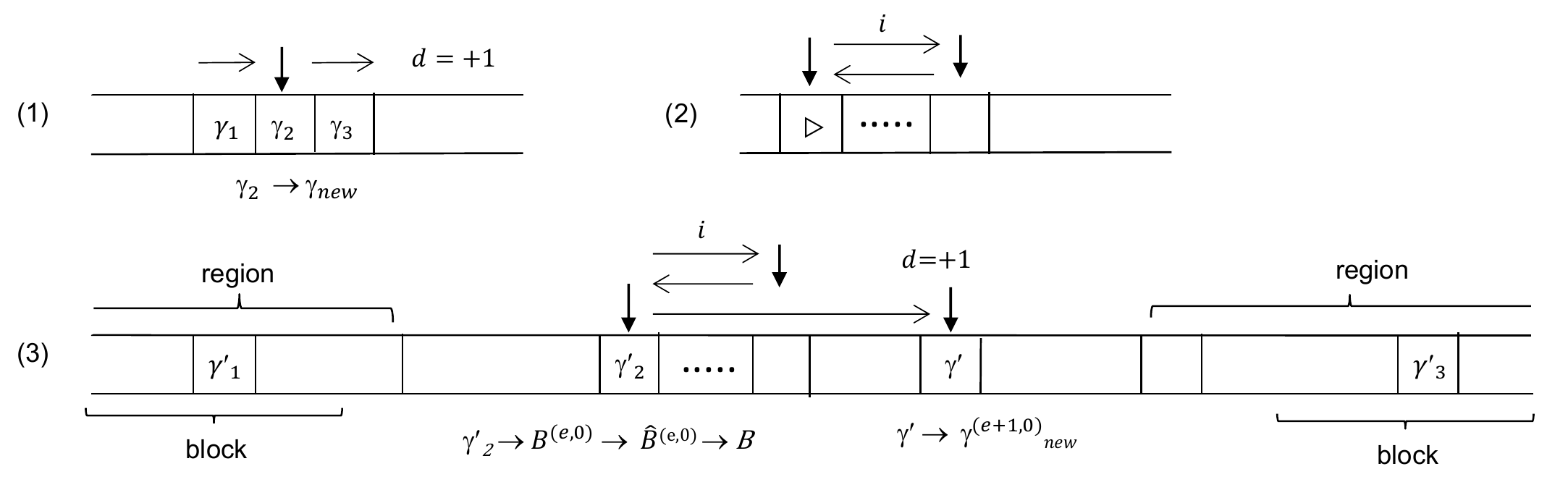}
\caption{A simulation of the movement of the storage-tape head and the counter head of $M$.
(1) After moving from $\gamma_1$ to $\gamma_2$, an input-tape-head of $M$  changes $\gamma_2$ to $\gamma_{new}$ and moves in direction $d=+1$.
(2) The counting head starts at cell $0$, travels through $i$ cells, and returns to cell $0$.
(3) The storage-tape head of $N$ moves as depicted in this figure to simulate (1) and (2).}\label{fig:storage-head-move}
\end{figure}


We lessen the number of input-tape heads from $\ell+2$ to $\ell+1$ for any $\ell\geq1$ by implementing a ``counter head'' to measure how far the tape head is away from a particular tape cell.
A \emph{counter head} is simply a two-way depth-susceptible tape head moving on an input tape in such a way that, once this tape head is activated, it starts moving from $\triangleright$ to the right, stays still for a while,  and comes back to $\triangleright$ with non-stopping movement (except for the period of reaching the rightmost location) by ignoring all input symbols on its way. We also require the storage-tape head to stay still during the activation of this counter head.

We partition every tape cell into two blocks, called the \emph{upper track} and the \emph{lower track}. The content of each tape cell is expressed by the track notation as $\track{\sigma}{\tau}$ if two tape symbols $\sigma$ and $\tau$ are written in those tracks \cite{TYL10}.

\begin{lemma}\label{half-reduction}
Let $k\geq2$ and $\ell\geq1$. Given any $k$-sda$_{2}$($\ell+2$) $M$ with a counter head over input alphabet $\Sigma$ running in polynomial time, there exists a polynomial-time $k$-sda$_{2}$($\ell+1$) with a counter head that recognizes the language $L_{ab}=\{ \tilde{x} \mid x\in L(M)\}$, where $a$ and $b$ are tape symbols not in $\Sigma$,  where
$\tilde{x} = (a\tilde{x}_1b)(a\tilde{x}_2b)\cdots (a\tilde{x}_nb)$ with $\tilde{x}_i = \track{x_1}{x_i}\track{x_2}{x_i}\cdots \track{x_n}{x_i}$ for $x=x_1x_2\cdots x_n$.
\end{lemma}


\begin{proof}
We fix $k\geq2$ and $\ell\geq1$. Let $M$ denote any polynomial-time $k$-sda$_{2}$($\ell+2$) with a counter head. Among all read-only tape heads indexed from $1$ to $\ell$, we call head 1 the principal-tape head and call head 2 and head 3 subordinate tape heads.
We wish to simulate these three tape heads  by two tape heads, eliminating one subordinate tape head. We leave all the remaining tape heads ``unmodified''  in the sense that they work only over $\tilde{x}_1$ in $a\tilde{x}_1b$ similarly to the original tape heads working over $x$.
This simulation can be carried out on an appropriate polynomial-time $k$-sda$_{2}$($\ell+1$), say, $N$.
Including $a$ and $b$, we set $t=|axb|$. The associated input to $N$ is $\triangleright\, \tilde{x}\, \triangleleft$.
Initially, heads $1$--$3$ of $M$ are all stationed at cell $0$.

We intend to simulate each step of $M$ by conducting a series of steps of $N$.
Assume that head $2$ is currently located at cell $i$ and head $3$ is at cell $j$. Such a location pair is expressed as $(i,j)$.
For any index $i\in[t]$, we call the block $a\tilde{x}_ib$ of the input $\tilde{x}$ as block $i$.
For convenience, $\triangleright$ and $\triangleleft$ are respectively called block 0 and block $n+1$.
We try to ``express'' the pair $(i,j)$ by stationing a designated tape head of $N$, say, head $h$, at the $i$th symbol of the $j$th block of $N$'s input tape.
Hereafter, we assume that head $h$ is currently located at the $i$th symbol of the $j$th block. We force $N$ to remember two input symbols, say, $(\sigma_i,\sigma_j)$ written at the $i$th and the $j$th cells of $M$'s input tape.

We further assume that, in a single step, heads $2$ and $3$ of $M$ move from $(i,j)$ to a new location pair $(i+d_1,j+d_2)$ by the given head directions  $d_1,d_2\in\{-1,0,+1\}$.
To simulate this single step of $M$,
$N$ needs to move head $h$ to a new location and read two symbols $(\sigma_{i+d_1},\sigma_{j+d_2})$.
If $d_2=0$, then $N$ moves head $h$ in direction $d_1$.
Otherwise, $N$ moves head $h$ in direction $d_1$, reads its input symbol $\sigma_{i+d_1}$, and remember it using the inner state.
Next, $N$ moves head $h$ leftward to the nearest $a$. As the head moves, $N$ also moves the counter head to the right (from the start cell) to count the number $i+d_1$.
Furthermore, $N$ moves head $h$ to the symbol $a$ in block $j+d_2$ without moving the counter head.
We restart the counter head backwardly. Finally, $N$ moves head $h$ rightward together with the counter head until the counter head comes back to the start cell.
At the time when the counter head arrives at the start cell,
head $h$ reaches the $(i+d_1)$th cell of block $j+d_2$.
Head $h$ then reads the input symbol $\sigma_{j+d_2}$ and $N$ remembers  $(\sigma_{i+d_1},\sigma_{j+d_2})$ using its inner state.
Note that the tape head on the storage tape never moves during the above process.

After the simulation of $M$, $N$ clearly reaches the same type (accepting or rejecting) of a halting state as $M$ does.
\end{proof}


Notice that a $k$-sda owns a one-way input-tape head whereas a $k$-sda$_{2}$($2$) uses a two-way input-tape head.
In the next lemma, we need to simulate  two-way head moves of a $k$-sda$_{2}$($2$) using one-way head moves of another $k$-sda$_{2}$(2).
For this purpose, we utilize a ``counter'' again together with the use of the reverse of an input.

\begin{lemma}\label{reversal-treat}
Given a polynomial-time $k$-sda$_{2}$($2$) $M$ with a counter head, there exists another $k$-sda$_{2}$($2$) $N$ with a counter head such that (i) $N$'s input-tape head never moves to the left and (ii) $N$ recognizes the language $L_{rev}= \{\hat{x} \mid x\in L(M)\}$ in polynomial time, where $\hat{x} = (a x\#x^Rb)^{|axb|}$ and $\#$ denotes a new separator.
\end{lemma}

\begin{proof}
Let $M$ be any polynomial-time $k$-sda$_{2}$($2$) with a counter head. We simulate the two-way  movement of $M$'s input-tape head by a one-way tape head in the following way.
Assume that $i$ represents the cell position of $M$'s input-tape head. Let $\sigma_i$ and $\sigma_{i-1}$ respectively denote the tape symbols at cells $i$ and $i-1$. We briefly call by head 1 $M$'s input-tape head other than the counter head.
If head 1 moves to the right or makes a stationary move, then $N$ simulates the step of $M$ exactly.
In what follows, we consider the case where head 1 moves to the left; that is, the new position of head 1 is $i-1$. By the depth-susceptibility condition of $M$, the current storage-tape cell contains no symbol in $\Gamma^{(k-1)}\cup \Gamma^{(k)}$ because, otherwise, head 1 cannot move.
In this case,
we move $N$'s input-tape head and the counter head simultaneously to the right until the tape head reaches the first encounter of $b$ in $ax\# x^Rb$.  We then continue moving the tape head rightward while we move the counter head back to the left endmarker. We finally make the tape head shift to the right cell to reach the symbol $\sigma_{i-1}$ in $ax\#$.
\end{proof}


Next, we show how to eliminate a counter head using the fact that the counter head is depth-susceptible.

\begin{lemma}\label{counter-head}
Let $M$ denote any polynomial-time $k$-sda$_{2}$($2$) $M$ with a one-way input-tape head and a counter head. There exists a depth-susceptible $k$-sda $N$ that recognizes $L' =\{\bar{x} \mid x \in L(M)\}$ in polynomial time, where $\bar{x}= x_1 1^n x_2 1^n \cdots 1^n x_n$ for $x=x_1x_2\cdots x_n$.
\end{lemma}

Sudborough made a similar claim, whose proof relies on Galil's argument \cite{Gal77}, which uses a stack to store and remove specific symbols in order to remember the distance of a tape-head location from a particular input tape cell. However, since tape cells on the storage tape are not allowed to modify more than $k$ times, we need to develop a different strategy to prove Lemma \ref{counter-head}.
For this purpose, we run an additional procedure of making enough open space on the depth-$k$ storage tape for future recording of the tape-head's movement.


\begin{proofof}{Lemma \ref{counter-head}}
Take any polynomial-time $k$-sda$_{2}$($2$) $M$ of the form $(Q,\Sigma, \{\Gamma^{(e)}\}_{e\in[0,k]_{\integer}}, \{\triangleright, \triangleleft\}, \delta,q_0,Q_{acc},Q_{rej})$ with a counter head.
We wish to simulate the behavior of the counter head using a storage tape of the desired $k$-sda $N$. In what follows, we describe this simulation procedure. Let $x$ denote any input and set $n=|\:{\triangleright}\,{x}\,{\triangleleft}\:|$. To simplify the description of the intended simulation by a $k$-sda $N$,
we deliberately allow $N$ to change any symbol not in $\Gamma^{(k)}$ into $B$ in a single step.

Recall that, whenever the counter head is activated, it starts at $\triangleright$, moves rightward for a certain number of steps, say, $i$, and moves back to the start cell to complete a single round of the ``counting'' process. During this process,
the storage-tape head stays still.
Since we force $N$ to move its  input-tape head exactly in the same direction as $M$'s, in what follows, we focus only on the simulation of $M$'s counter and storage-tape heads.
Using its inner states, $N$ can remember (a) which direction the storage-tape head comes from  and (b) the contents of  the tape cell currently scanned and its left and right adjacent tape cells (if any).

For convenience, we set $\gamma^{(k,0)}=\gamma^{(0,k)}=B$ and $B^{(k)}=\hat{B}^{(k)}=B$. We introduce special storage alphabets $\Gamma^{(j)}_N$ ($j\in[0,k]_{\integer}$) for $N$.
Firstly, we set $\Gamma^{(0)}_N = \Gamma^{(0)}$ and $\Gamma^{(k)}_N = \Gamma^{(k)}$. For any $j\in[k-1]$, the storage alphabet $\Gamma^{(j)}_N$ in part consists of symbols of the form $\gamma^{(0,e)}$ and $\gamma^{(e,0)}$ for $e\in[0,k]_{\integer}$ associated with $\gamma$ in $\Gamma^{(j)}$. Moreover, $\Gamma^{(j)}_N$ contains symbols $B^{(c)}$ and $\hat{B}^{(c)}$ for all $c\in[0,k]_{\integer}$.

We partition the storage tape into a number of ``regions'', where a  \emph{region} consists of $2k$ blocks, each of
which contains $n$ tape cells.
Each region is meant to simulate one run of the counter head and it basically holds the information on one storage symbol. Two regions are separated by one special ``separator block'' of $n$ cells. See an illustration in Figure~\ref{fig:storage-head-move} for blocks and regions.
Each block contains a string, which has one of the following three forms:  $\gamma' B_1B_2\cdots B_{n-1}$ with $\gamma' \in \Gamma_N - \{\Box,\triangleright\}$, $B_1B_2\cdots B_{n-1}\gamma'$ with $\gamma' \in \Gamma_N - \{\triangleright\}$, and $B_1B_2\cdots B_n$, where $B_i \in\{B^{(c)},\hat{B}^{(c)}\mid c\in [0,k-1]_{\integer}\}\cup\{B\}$ for any $i\in[n]$. If $\gamma'\neq B$, then $\gamma'$ is called a \emph{representative} of a block.
A tape cell in a block is also called a \emph{representative} if it contains a representative of the block.
A block is called \emph{active} if it contains a representative, and all other blocks are called \emph{passive}. In particular, we call a block \emph{consumed} if it is filled with $B$, and thus there is no representative. We say that a block is \emph{blank} if the block consists only of $B$.
The parameter $e$ in $\gamma^{(0,e)}$ (resp., $\gamma^{(e,0)}$) indicates the existence of $e$ consumed blocks in the area of the region that are left (resp., right) to the currently scanning cell.

In a run of the procedure described below, we maintain the circumstances, in   which there is at most one active block in each region.
Assume that $\gamma_1\gamma_2\gamma_3$ is the content of three neighboring cells, the middle of which is being scanned by $M$'s storage-tape head, provided that, whenever $\gamma_2$ equals $\triangleright$, we automatically ignore $\gamma_1$.
Assuming that $\gamma'_1$, $\gamma'_2$, and $\gamma'_3$ are three  representatives associated respectively with $\gamma_1$, $\gamma_2$, and $\gamma_3$.
Assume that $M$'s tape head stays over $\gamma_2$ and $N$'s tape head does over $\gamma'_2$.


Let us focus on a computation of $M$ on input $x$ and simulate this  computation on the $k$-sda $N$. To this simulation easier, we partition the computation into a series of ``session steps,'' each of which constitutes a consecutive sequence of $M$'s moves defined as follows. Assume that, at time $t$, the storage-tape head of $M$ is resting on a cell containing a non-$B$ symbol, say, $\gamma_0$.
We consider a consecutive cells holding a string of the form $\gamma_{-t_1}\gamma_{-t_1+1}\cdots \gamma_{-1} \gamma_0 \gamma_1 \cdots \gamma_{t_2-1}\gamma_{t_2}$ satisfying that $t_1,t_2\in\nat$, $\gamma_{-t_1}\notin \{\Box,B\}$, $\gamma_{t_2}\neq B$, and $B_i=B$ for all $i\in[-t_1+1,t_2-1]_{\integer}-\{0\}$.
As an example, at time $0$ when $M$ scans $\triangleright$ at the start cell and the rest of the cells are $\Box$, we obtain $t_1=0$, $t_2=1$, $\gamma_0=\triangleright$, and $\gamma_{t_2}=\Box$.
A \emph{session step} consists of one of the following series of moves of $M$: starting at $\gamma_0$, (1) write $\gamma_{new}$ over $\gamma_0$ and make a stationary move, (2) write $\gamma_{new}$, move in direction $d\in\{-1,+1\}$ passing through $\gamma_i$'s, reach either $\gamma_{-t_1}$ or $\gamma_{t_2}$, and stop, and (3) write $\gamma_{new}$, make one step in direction $d\in\{-1,0,+1\}$, make a turn, pass through $\gamma_i$'s, reach either $\gamma_{-t_1}$ or $\gamma_{t_2}$, and stop.

In the following simulation, each symbol $\gamma_i$ is translated into a region that contains its corresponding symbol $\gamma'_i$.


(*) Meanwhile, we assume that $\gamma'_0$ is of the form $\gamma_0^{(e,0)}$ for a certain $e\in[0,k]_{\integer}$ and that $\gamma'_0$ is written at cell, say, $s\in \nat$. The first assumption implies the existence of $e$ consecutive blank blocks in the left-side area of $\gamma_0^{(e,0)}$.
It also follows that the current block has the form $\gamma'_0 B^{(c)}\cdots B^{(c)}$ for a certain constant $c\in [0,dv(\gamma_0)-1]_{\integer}$.
In what follows, we also assume that the previous session step of $M$ is not a storage-stationary move. We argue two cases (I)--(II) separately.

(I) Consider the first case where the counter head is not activated.
Since we do not need to simulate the behavior of the counter head, it suffices to simulate the next session step of $M$ using $N$'s storage tape.

(1) Consider the case where the storage-tape head of $M$ came from the left. Clearly, $\gamma_0\neq\triangleright$ follows. Assume that we remember the symbol $\gamma'_{-t_1}$ in the form of inner state. In what follows, we examine two cases, depending on whether $\gamma_0$ is $\Box$ or not.


(i) Assume that $\gamma_0\neq \Box$. Note that both separator blocks around $\gamma'_0$ in the current region are not blank. This makes it possible to discern the borders to the neighboring regions.

(a) If $d=0$, then $N$ changes $\gamma_0^{(e,0)}$ to $\gamma^{(e,0)}_{new}$ without further moving its storage-tape head.

(b) Assume that $d=+1$. There are two more cases to consider.
If $M$ moves rightward to $\gamma_{t_2}$, then $N$ overwrites $\gamma_0^{(e,0)}$ by $\gamma_{new}^{(e,0)}$, moves rightward by changing $B^{(c)}$ to $B^{(c+1)}$, and crosses the border to the neighboring region.
From there, $N$ skips blank regions until entering the first non-blank region, which holds $\gamma'_{t_2}$.  There may be a ceratin number of blank blocks in this region before reaching $\gamma'_{t_2}$.
Since $\gamma_{t_2}\neq B$, $N$ can find this representative $\gamma'_{t_2}$ in this neighboring region and stops at the cell containing $\gamma'_{t_2}$.
In contrast, if $M$ steps to $\gamma_1$ ($=B$) and then returns to $\gamma_{new}$, then $N$ writes $\gamma_{new}^{(e,0)}$ over $\gamma_0^{(e,0)}$, moves to the right, crosses the border to the neighboring region, makes a turn at the time of encountering the first $B$, returns to $\gamma_{new}^{(e,0)}$. If $dv(\gamma_0)<k-1$, then $N$ stops here; otherwise, $N$ continues moving leftward until reaching $\gamma'_{-t_1}$, and then stops.

(c) When $d=-1$, $N$ writes $\gamma_{new}^{(e,0)}$ over $\gamma_0^{(e,0)}$, makes a left turn, continues moving leftward until finding a representative $\gamma'_{-t_1}$, and stops.


(ii) Consider the case where $\gamma_0=\Box$. Note that all cells located in the  right-side area  of cell $s$ hold $\Box$. There are three cases (a)--(c) to examine.

(a) If $d=0$, then $N$ writes $\gamma^{(0,0)}_{new}$ over $\Box$ and makes a stationary move. This is possible because $\gamma_{new}^{(0,0)}\neq B$.

(b) In contrast, assume that $d=+1$. We write $\gamma_{new}^{(0,0)}$ over $\Box$ and move to the right. Since we need to secure enough open space for future simulations of the counting head, we wish to generate a new region. In an early simulation, we have already generated the first half of a region, and thus we need to generate the second half of this region at first.
For this purpose, using $1^n$ on the input tape, $N$ moves rightward passing through $(k-e-1)n$ cells by changing $\Box$ to $B^{(1)}$, creates a new border, continues moving for $kn$ cells by changing $\Box$ to $B^{(1)}$ to find the center of the new region, and finally stops. This last process newly generates the first half of a region.

(c) Finally, when $d=-1$, $N$ writes $\gamma_{new}^{(0,0)}$ over $\Box$, makes a left turn, crosses the first border to the neighboring blank region, continues skipping blank regions until entering the region containing $\gamma'_{-t_1}$, and then stops.


(2) Consider the case where the storage-tape head came from the right. The major deviation from (1) is the case of $\gamma_0=\triangleright$. In this case, if $d=0$, then $N$ makes a stationary move. If $d=+1$, then $N$ crosses the border to the right neighboring region, continues skipping blank regions until finding $\gamma'_{t_2}$, and then stops. All the other cases are handled symmetrically to (1).


(II) Next, we want to simulate a single counting process by the counter head of $M$ on the depth-susceptible $k$-sda $N$. Assume that the counting head travels rightward passing through $i$ cells and then returns to the start cell. There are three cases (1)--(3) to consider separately. Note that $\gamma_0\neq \triangleright$ since $\triangleright\in\Gamma^{(k)}$.

(1) Consider the case where the storage-tape head of $M$ came from the left. This implies that $\gamma_0\neq\triangleright$.
In this case, we need to mimic the back-and-forth movement of the counter head as follows.
The machine $N$ remembers $\gamma'_0$ in the form of inner states, modifies it to $\hat{B}^{(a+1)}$ ($\notin \Gamma^{(k)}$) with $a=dv(\gamma_0)$,  and moves its tape head for $i$ steps to the right as the counter head does, by changing every encountered symbol of the form $B^{(c)}$ with $c\in[0,k]_{\integer}$ to $B^{(\min\{c+1,k\})}$ on its way,  provided that $B^{(0)}$ denotes $\Box$.
After making $i$ steps, it makes a left turn, returns to $\hat{B}^{(a+1)}$, and writes $B$ over $\hat{B}^{(a+1)}$.
The storage-tape head of $N$ again starts moving rightward for exactly $n$ steps (by reading $1^n$ on the input tape) by changing each
$B^{(\min\{c+1,k\})}$ to $B^{(\min\{c+2,k\})}$, and it finally writes $\gamma_{new}^{(e+1,0)}$ since we create an additional consumed block in the left-side area of $\gamma^{(e+1,0)}_{new}$.

Finally, if $d=0$, then the tape head stops here. By contrast, when $d=+1$ (resp., $d=-1$),  $N$ moves the storage-tape head rightward (resp., leftward), skipping blank regions until encountering a representative $\gamma'_{t_2}$ (resp, $\gamma'_{-t_1}$).


(2) Assume that $M$ moved to $\gamma_{2}$ from the right. Note that  $\gamma_0\neq\Box$. Symmetrically to (1), we generate a new consumed block in the left-side area of $\gamma'_0$.


(**) To complete the simulation, let us consider the second case where
$\gamma'_0$ has the form $\gamma^{(0,e)}_0$ for a certain number $e\in[0,k]_{\integer}$. The current block has the form $B^{(c)}\cdots B^{(c)}\gamma'_0$ and there are $e$ consumed blocks in the right-side area of $\gamma^{(0,e)}_2$. Nonetheless, this case can be symmetrically treated by skipping all consumed blocks as described above.
\end{proofof}


Finally, we combine all the lemmas (Lemmas \ref{two-input-model}--\ref{reversal-treat}) and  verify Theorem \ref{auxiliary-characterize}.

\vs{-2}
\begin{proofof}{Theorem \ref{auxiliary-characterize}}
Let $k\geq2$.
The implication (1)$\Rightarrow$(2) is shown as follows.
Take any language $L$  over alphabet $\Sigma$ in $\mathrm{LOG}k\mathrm{SDA}$.
There exist a function $f:\Sigma^*\to\Sigma_2^*$ in $\fl$ for an appropriate alphabet $\Sigma_2$ and a depth-susceptible $k$-sda $M$ working over $\Sigma_2$ such that, for any string $x\in\Sigma^*$, if $x\in L$, then $M$ accepts the string $f(x)$; otherwise, $M$ rejects it. By Lemma \ref{two-input-model}, we can obtain a log-space depth-susceptible aux-$k$-sda $N$ that recognizes $\{x\in\Sigma^*\mid f(x)\in L(M)\}$ in polynomial time.

Lemma \ref{aux-to-6-heads} obviously leads to the implication  (2)$\Rightarrow$(3). Finally, we want to show that (3) implies (1).
Given a language $L$, we assume that there is a polynomial-time $k$-sda$_2(\ell)$ $M$ recognizing $L$ for a certain number $\ell\geq2$.
We transform this $k$-sda$_{2}$($\ell$) to another $k$-sda$_{2}$($\ell+1$) $M'$ by providing a (dummy) counter head.
We repeatedly apply Lemma \ref{half-reduction} to reduce the number of input-tape heads down to $2$ by modifying the target language to $L_{ab}$.
Lemma \ref{reversal-treat} then implies the existence of a $k$-sda$_{2}$($1$) $N$ with one-way input-tape and counter heads such that $N$ correctly recognizes $L_{rev}$ in polynomial time.
By Lemma \ref{counter-head}, we further obtain a polynomial-time depth-susceptible $k$-sda $K$ that can recognize $L'$,  which is of the form $\{\bar{x}\mid x\in L_{ab}\}$.
Given an input $z$ to $L$, we define $f(z)$ to be the input $\bar{x}$  obtained by running a series of the processes, of Lemmas \ref{half-reduction}--\ref{reversal-treat}, which reduce the number of input-tape heads.
By the clear description of these reduction processes, this function $f$ is computed in polynomial-time using only log space. Since $L'=\{f(x)\mid x\in L\}$ and $f\in\fl$, we conclude that
$L$ belongs to $\mathrm{LOG}k\mathrm{SDA}$.
\end{proofof}


\section{Universal Simulators and L-m-Hard  Languages}\label{sec:hardest-language}

As a major characteristic feature, we intend to prove the existence of concrete, generic $\dl$-m-hard languages for $\mathrm{LOG}k\mathrm{SDA}$ for each index $k\geq2$. For this purpose, we first construct a universal simulator that has an ability to precisely simulate all depth-susceptible $k$-sda's when  appropriate encodings of both $k$-sda's and inputs are given. We further force this universal simulator to be a ``depth-immune'' $k$-sda$_{2}(4)$.

\subsection{LOG$k$SDA-Hard Languages}

Sudborough \cite{Sud78} earlier proposed, for every number $m\geq2$, the special ``tape-hardest'' language $L_0^{(m)}$, which is $\dl$-m-hard for $\mathrm{DCFL}$. This language literally encodes all transitions of deterministic pushdown automata so that we can simulate these machines step by step using a stack.
Since $L^{(m)}_0$ belongs to $\mathrm{LOGDCFL}$, it is also $\dl$-m-complete for $\mathrm{LOGDCFL}$ because $\mathrm{LOGDCFL}$ is closed under $\dl$-m-reductions.
Sudborough's success comes from the fact that the use of one-way and two-way deterministic pushdown automata makes no difference in formulating  $\mathrm{LOGDCFL}$.
In a similar spirit, we propose the following decision problem,  $\mathrm{MEMB}_k$, for $k\mathrm{SDA}$.
Recall that a decision problem is identified with its associated language.

\ms
{\sc Membership $k$SDA Problem} ({\sc MEMB$_k$}):
\renewcommand{\labelitemi}{$\circ$}
\begin{itemize}\vs{-2}
  \setlength{\topsep}{-2mm}%
  \setlength{\itemsep}{0mm}%
  \setlength{\parskip}{0cm}%

\item {\sc Instance:} an encoding $\pair{M,x}$ of a depth-susceptible $k$-sda $M$ over alphabet $\Sigma$ and an input $x$.

\item {\sc Question:} does $M$ accept $x$?
\end{itemize}

A key to an introduction of $\mathrm{MEMB}_k$ is a ``generic'' scheme of encoding both a depth-susceptible $k$-sda $M$ and an input $x$ into a single string $\pair{M,x}$ over a fixed alphabet, which is independent of the choice of $M$ and $x$. In Section \ref{sec:membership}, we will explain such an encoding scheme in detail.

Let us recall the depth-susceptibility condition imposed on $k$-sda$_{2}$($\ell$) in Section \ref{sec:multi-head}, which is a restriction on the behavior of the $k$-sda$_{2}$($\ell$) while reading storage-tape symbols in $\Gamma^{(k-1)}\cup\Gamma^{(k)}$.
We remove this condition and introduce a \emph{depth-immune $k$-sda$_{2}$($\ell$)} in a way similar to the introduction of a depth-immune $k$-sda in Section \ref{sec:TMs}.
With this new model, we define $k\mathrm{SDA}_{2}(\ell)$ to be the family of   all languages recognized by depth-immune $k$-sda$_{2}$($\ell$)'s running in polynomial time.

We assert the following two statements.

\begin{theorem}\label{hardness-prop}
Let $k\geq2$.
\renewcommand{\labelitemi}{$\circ$}
\begin{enumerate}\vs{-2}
  \setlength{\topsep}{-2mm}%
  \setlength{\itemsep}{0mm}%
  \setlength{\parskip}{0cm}%

\item The language $\mathrm{MEMB}_k$ belongs to $k\mathrm{SDA}_{2}(4)$.

\item $\mathrm{MEMB}_k$ is $\dl$-m-hard for $k\mathrm{SDA}$ and thus for $\mathrm{LOG}k\mathrm{SDA}$.
\end{enumerate}
\end{theorem}

To prove Theorem \ref{hardness-prop}(1), it suffices in essence to verify the following lemma. For convenience, a $k$-sda$_2$($\ell$) $U$ is called a \emph{universal simulator} for depth-susceptible $k$-sda's if, for any depth-susceptible $k$-sda $M$ and any input $x$ given to $M$, (i) $U$ takes an input of the form $\pair{M,x}$ and (ii) if $M$ halts on $x$, then $U$ enters the same type (accepting or rejecting) of halting states as $M$ does; otherwise, $U$ does not halt.

\begin{lemma}\label{universal-simulator}
For each $k\geq2$, there exists a special depth-immune $k$-sda$_{2}$($4$) that works as a universal simulator for depth-susceptible $k$-sda's.
\end{lemma}


In Section \ref{sec:proofs-main}, we wish to prove Lemma  \ref{universal-simulator} and subsequently Theorem \ref{hardness-prop}. The proof of the lemma provides a detailed construction of the desired $k$-sda$_{2}$($4$) universal simulator, which takes an encoding $\pair{M,x}$ of a depth-susceptible $k$-sda $M$ and its input $x$ and then properly simulates $M$ on $x$ using only a depth-$k$ storage tape of the simulator.
Such a proper encoding scheme enables us to construct the desired universal simulator.
Theorem \ref{hardness-prop}(1) instantly follows from the existence of a universal simulator.

\subsection{A Desirable Encoding Scheme}\label{sec:membership}

Hereafter, we describe the desired encoding $\pair{M,x}$ of a depth-susceptible $k$-sda $M$ and an input $x$. We need to heed special attention to how to encode a pair of $M$ and $x$ into a single string $\pair{M,x}$ so that we can easily retrieve $M$ and $x$ from $\pair{M,x}$ using a depth-$k$ storage tape  for the purpose of the simulation of $M$ on $x$.
The desired encoding of $M$ and $x$ needs to keep all the information on the transitions of $M$ intertwined with all bits of $x$ in sequence.
However, since the storage usage of $k$-sda's are quite different from that of deterministic pushdown automata, our encoding scheme is therefore quite different from Sudborough's scheme.


An underlying idea of Sudborough's construction of his tape-hardest languages is the notion of \emph{cancelling pairs}. A similar idea will be used implicitly in the following construction.


Since a $k$-sda uses arbitrary sets $Q$, $\Sigma$, and $\Gamma$, we need to express them using only fixed alphabets independent of $M$ and $x$.
In the rest of this subsection, we assume that $M$ has the form $(Q_M,\Sigma_M, \{\Gamma^{(e)}_M\}_{e\in[0,k]_{\integer}}, \{\triangleright,\triangleleft \}, \delta,p_1,Q_{M,acc},Q_{M,rej})$ satisfying the following specific conditions: $Q_M=\{p_1,p_2,\ldots,p_{m_1}\}$, $\check{\Sigma}_{M} = \Sigma_{M}\cup\{{\triangleright,\triangleleft}\} = \{\sigma_1,\sigma_2,\ldots,\sigma_c\}$, and
$\Gamma_{M} =\bigcup_{e\in[0,k]_{\integer}}\Gamma^{(e)}_{M} = \{\gamma_1,\gamma_2,\ldots,\gamma_{m_2}\}$ with $\Gamma^{(0)}_{M}=\{\Box\}$ and $\Gamma^{(k)}_{M}=\{\triangleright,B\}$.
Without loss of generality, we further assume that $Q_{M,acc} = \{p_2\}$ and $Q_{M,rej}=\{p_3\}$.
The transition function $\delta$ thus maps $Q_M\times \check{\Sigma}_{M}\times \Gamma_{M}$ to $Q_M\times \Gamma_M\times D_1\times D_2$.
For simplicity, we assume that $c$, $m_1$, and $m_2$ satisfy that $c=2^{\ceilings{\log{c}}}$, $m_1=2^{\ceilings{\log{m_1}}}$, and $m_2=2^{\ceilings{\log{m_2}}}$. Under this assumption, all the elements in $\check{\Sigma}_{M}$, $Q_M$,  and $\Gamma_{M}$ are precisely expressed as the corresponding elements in $\{a_0,a_1\}^{\ceilings{\log{c}}}$, $\{a_0,a_1\}^{\ceilings{\log{m_1}}}$, and $\{a_0,a_1\}^{\ceilings{\log{m_2}}}$ (using their lexicographic order), respectively.

In what follows, let $e,e'\in[0,k]_{\integer}$, $h_1,h_3\in[m_1]$, $h_2,h_4\in[m_2]$, $h_5\in[c]$, $\gamma_{h_1}\in\Gamma^{(e)}_{M}$, $\gamma_{h_3}\in\Gamma^{(e')}_{M}$, and $d_1\in\{0,+1\}$.
Notice that the input tape is read by $M$ only once from left to right.
Generally, a transition of $M$ has one of the following 7 forms.

\s

(1) $\delta(p_{h_2},\sigma_{h_5},\gamma_{h_1}) = (p_{h_4},\gamma_{h_3},d_1,+1)$ (moving to the right) if $e<k$ is even and $e'= \min\{e+1,k\}$.

(2) $\delta(p_{h_2},\sigma_{h_5},\gamma_{h_1}) = (p_{h_4},\gamma_{h_3},d_1,-1)$ (left turn) if $e<k$ is even and $e'= \min\{e+2,k\}$.

(3) $\delta(p_{h_2},\sigma_{h_5},\gamma_{h_1}) = (p_{h_4},\gamma_{h_3},d_1,-1)$ (moving to the left) if $e<k$ is odd and $e'=\min\{e+1,k\}$.

(4) $\delta(p_{h_2},\sigma_{h_5},\gamma_{h_1}) = (p_{h_4},\gamma_{h_3},d_1,+1)$ (right turn) if $e<k$ is odd and $e'=\min\{e+2,k\}$.

(5) $\delta(p_{h_2},\sigma_{h_5},\triangleright) = (p_{h_4},\triangleright,d_1,+1)$.

(6) $\delta(p_{h_2},\sigma_{h_5},\gamma_{h_1}) = (p_{h_4},B,0,d_2)$ if $d_2\neq0$ and either $e=k$ or $e=k-1$.

(7) $\delta(p_{h_2},\sigma_{h_5},\gamma_{h_1}) = (p_{h_4},\gamma_{h_1},d_1,0)$ (storage-stationary move).

\s

We further set $m=(m_1,m_2)$ and  $\theta =m_1m_2$ ($=|Q_M||\Gamma_M|$).
Moreover, we fix two bijections mapping $[m_1]\times [m_2]$ to $[\theta]$ and mapping $[m_1]\times [m_2]\times [m_1]\times [m_2]$ to $[\theta^2]$.
For the desired universal simulator, we define a distinguished symbol ``$\#$'' and we define $\tilde{\Sigma}$ and $\tilde{\Gamma}^{(e)}$, for every  $e\in[0,k]_{\integer}$, to satisfy $\tilde{\Sigma} =\{a_0,a_1,0,1,\#\}$,  $\tilde{\Gamma}^{(0)}=\{\Box\}$, $\tilde{\Gamma}^{(1)}=\{a_0,a_1,0,1,\#\}$, and $\tilde{\Gamma}^{(k)}=\{\triangleright,B\}$.
Finally, we define $\tilde{\Gamma}$ to be $\bigcup_{e\in[0,k]_{\integer}}\tilde{\Gamma}^{(e)}$.

We encode an input $x=x_1x_2\cdots x_n$ of length $n$ symbol by symbol as follows. Notice that $x_0=\triangleright$ and $x_{n+1}=\triangleleft$.
Let $X_{l} = x_l\#$ for any position $l\in[0,n+1]_{\integer}$, where $x_l$ is viewed here as a string over $\tilde{\Sigma}$.
These strings $X_l$ will be later combined with an encoding of transitions of $M$.

We then encode the transitions of the forms (1)--(7) in the following fashion. Since all information stored in a storage tape must be modified as a tape head traverses through the tape, we need to prepare multiple copies of all the transition rules.

Let $h=(h_{1},h_{2},h_3,h_4)$, $h'=(h_2,h_4)$, and $h''=(h_1,h_2,h_4)$, where $h_1,h_3\in[m_2]$ and $h_2,h_4\in[m_1]$. Moreover, let $h_5\in[c]$ and $d_1\in\{0,+1\}$.
By treating $h$ as an element of $[m_1]\times[m_2]\times[m_1]\times[m_2]$, we identify it with a number in $[\theta^2]$. Similarly, $h'$ is identified with a number in $[m_1^2]$.
Abusing the notations, we further assume that $h\in[\theta^2]$ and $h'\in[m_1^2]$.
We encode the transitions (1)--(7) into six types of strings $\tilde{u}_{h,e}^{(+)}$, $\tilde{u}_{h,e}^{(-)}$, $\tilde{u}_{h,e}^{(lt)}$,  $\tilde{u}_{h,e}^{(rt)}$, $\tilde{u}_{h,e}^{(0)}$, and $\tilde{u}_{h,e}^{(d)}$, which are defined below, where the superscripts ``$lt$'' and ``$rt$'' respectively refer to ``left turn'' and ``right turn.''
Those strings are later referred to as \emph{encoded transitions}.
Hereafter, for convenience, $p_n$, $\sigma_n$, and $\gamma_n$ are assumed to be already translated.
The items (1$'$)--(7$'$) below directly correspond to the transitions (1)--(7). For readability, nonetheless, we omit all the supplemental conditions stated in (1)--(7).
We set $\tilde{d}= 100,000,001$ if $d=-1,0,+1$, respectively.

\s

(1$'$) $\tilde{u}^{(+)}_{h,0} \equiv \sigma_{h_5}\# \gamma_{h_{1}}\# p_{h_2}\# 001 \#      (\gamma_{h_3}\#)^{\theta}\# p_{h_4}\# 1^{d_1+1}0^{1-d_1}\#$.

(2$'$) $\tilde{u}^{(lt)}_{h,0} \equiv \sigma_{h_5}\# \gamma_{h_{1}} \# p_{h_2} \# 100 \# (\gamma_{h_3}\#)^{\theta} \# p_{h_4} \# 1^{d_1+1}0^{1-d_1}\#$.

(3$'$) $\tilde{u}^{(-)}_{h,1} \equiv \sigma_{h_5}\# \gamma_{h_{1}} \# p_{h_2} \# 100 \#
(\gamma_{h_3}\#)^{\theta}\# p_{h_4} \# 1^{d_1+1}0^{1-d_1}\#$.

(4$'$) $\tilde{u}^{(rt)}_{h,1} \equiv \sigma_{h_5}\# \gamma_{h_{1}} \# p_{h_2} \# 001 \# (\gamma_{h_3}\#)^{\theta} \#  p_{h_4} \# 1^{d_1+1}0^{1-d_1}\#$.

(5$'$) $\tilde{u}^{(+)}_{h',2} \equiv \sigma_{h_5}\# \:{\triangleright} \:\# p_{h_2} \# 001 \# \:
{\triangleright} \: \# p_{h_4} \# 1^{d_1+1}0^{1-d_1} \#$.

(6$'$) $\tilde{u}^{(d_2)}_{h',2} \equiv \sigma_{h_5} \#
\gamma_{h_1} \# p_{h_2} \# \tilde{d_2} \# B^{\alpha} \#
p_{h_4} \# 10 \#$.

(7$'$) $\tilde{u}^{(0)}_{h',3} \equiv \sigma_{h_5}\# \gamma_{h_1} \# p_{h_2} \#  000 \# (\gamma_{h_1}\#)^{\theta}
\# p_{h_4} \# 1^{d_1+1}0^{1-d_1} \#$.

\s

Let us define $P^{(+)}$ as follows.
For any $j\in\{h,h',h''\}$ and any $e\in\{0,1,2,3\}$, the segment $\sigma_{h_5}\# \gamma_{h_1} \# p_{h_2} \#$ (as well as $\sigma_{h_5}\# {\triangleright}\, \# p_{h_2} \#$) of $\tilde{u}^{(+)}_{j,e}$ is called a \emph{receptor} of  $\tilde{u}^{(+)}_{j,e}$ and the segment $001 \# (\gamma_{h_3}\#)^{\theta}\# p_{h_4} \# 1^{d_1+1}0^{1-d_1}\#$ of $\tilde{u}^{(+)}_{j,e}$ is called a \emph{residue} of  $\tilde{u}^{(+)}_{j,e}$.
We set $\tilde{u}^{(+)}_{0}$ to be $\tilde{u}^{(+)}_{1,0} \# \tilde{u}^{(+)}_{2,0} \# \cdots \tilde{u}^{(+)}_{\theta^2,0}\#$ and  $\tilde{u}^{(+)}_2$ to be $\tilde{u}^{(+)}_{1,2} \# \tilde{u}^{(+)}_{2,2} \# \cdots \tilde{u}^{(+)}_{m_1^2,2}\#$.
Collectively, we set $P^{(+)}$ to be $\tilde{u}^{(+)}_0 \# \tilde{u}^{(+)}_2 \#$. The other three strings $P^{(lt)}$, $P^{(-)}$, and $P^{(rt)}$ are defined similarly. Combining those four strings, we define $P$ to be $P^{(+)} \# P^{(lt)}\# P^{(-)}\# P^{(rt)}\#$.

Let $S$ denote the string $p_1\# p_2\# \cdots p_{m_1}\#$ and let $G$ be $\gamma_1\# \gamma_2\# \cdots \gamma_{m_2}\#$. Furthermore, we set $X$ to be $X_0 X_1 \cdots X_{n+1}\#$.
Finally, the encoding $\pair{M,x}$ of $M$ and $x$ is defined as the string   $X\# S \# G \# P \#$.
Note that the length of $\pair{M,x}$ is bounded by $O((|Q_M||\Gamma_M||\check{\Sigma}_M||x|)^4)$. It is not difficult to see by the definition that the encoding $\pair{M,x}$ is uniquely determined from $M$ and $x$.

\subsection{Proofs of Theorem \ref{hardness-prop} and Lemma  \ref{universal-simulator}}\label{sec:proofs-main}

Our goal of this subsection is to provide the proofs of Theorem \ref{hardness-prop} and Lemma  \ref{universal-simulator}. For this purpose, we first present, given an arbitrary encoded string $\pair{M,x}$, how to simulate $M$ on $x$ precisely using a depth-$k$ storage tape of a universal simulator.

We first introduce necessary terminology.
A \emph{storage-tape content} refers to a sequence of symbols written on the storage tape from the start cell to the leftmost tape cell containing $\Box$, say, cell $t$.
For our convenience, we call a string $\xi$ of the form $0^{b}\sigma_0\sigma_1\cdots \sigma_{l-1}\hat{\sigma}_l
\sigma_{l+1}\cdots \sigma_{t}$ a \emph{storage configuration} of $M$ if $M$ is scanning cell $b$ of the input tape, the storage-tape head  of $M$ is scanning cell $l$, $\sigma_0$ equals $\triangleright$,   $\hat{\sigma}_{l}$ is of the form $(1^h,\sigma_l)$, and the string $\sigma_0\cdots \sigma_{l-1}\sigma_{l}\sigma_{l+1} \cdots \sigma_{t}$ is a storage-tape content of $M$.

Here is a key idea behind the following proof of Lemma \ref{universal-simulator}. A storage tape is used to keep track of the information on the current surface configuration. To simulate the next move, we go through all encoded transitions on an input tape by way of removing any ``cancelling pair'' to locate the target transition to be taken at the next move. If we successfully delete the old surface configuration, we write a new surface configuration into an ``appropriately chosen'' new section of the storage tape. This is needed because our storage tape is of depth $k$, and thus we cannot place a new symbol at an arbitrary location of the storage tape.


\begin{figure}[t]
\centering
\includegraphics*[height=4.2cm]{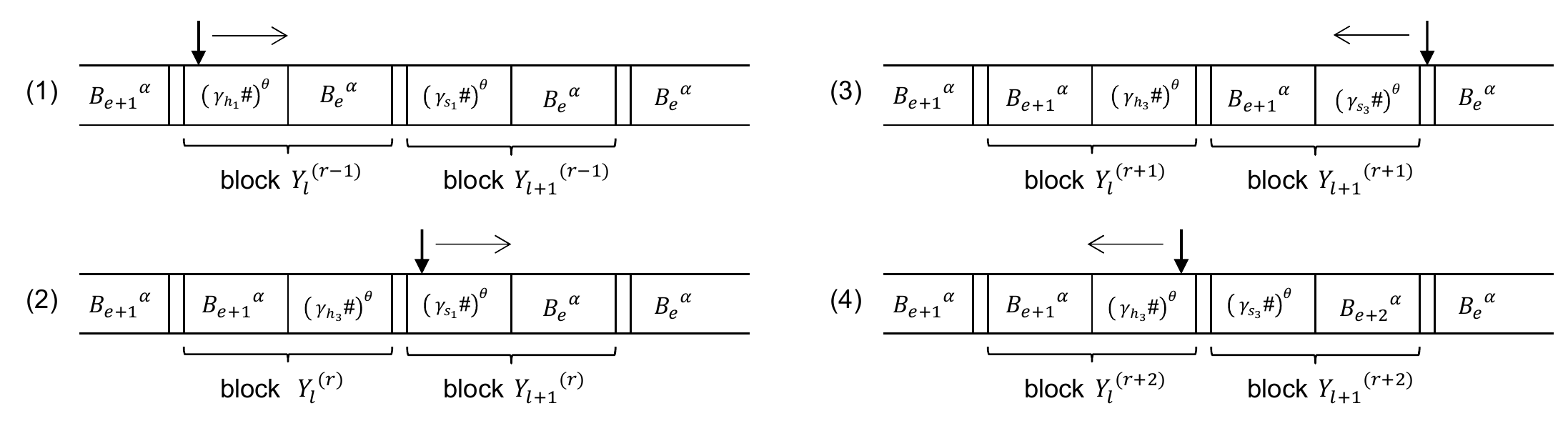}
\caption{Blocks on the storage tape and the tape head positions at steps $r-1$, $r$, and $r+1$, where $h_1,h_3,s_1,s_3\in[m_2]$ and $e+2<k$.
The tape head moves to the right in (1)--(2) and to the left in (3)--(4) by making a left turn in (3).}\label{fig:simulation-sda}
\end{figure}


\begin{proofof}{Lemma \ref{universal-simulator}}\label{sec:case-equal-2}
Let $k\geq2$. We intend to construct a depth-immune $k$-sda$_2$($4$) universal simulator working for all depth-susceptible $k$-sda's. Let $M = (Q_M,\Sigma_{M},\{\Gamma^{(e)}_M\}_{e\in[0,k]_{\integer}}, \{\triangleright, \triangleleft \}, \delta, p_1, Q_{M,acc}, Q_{M,rej})$ denote any depth-susceptible $k$-sda and let $x=x_1x_2\cdots x_n$ be any input of length $n$ given to $M$. Additionally, we set $x_0=\:\triangleright$ and $x_{n+1}=\:\triangleleft$.
Let us recall the notations, $\tilde{\Sigma}$, $\tilde{\Gamma}^{(e)}$, $\tilde{\Gamma}$, $c$, $m_1$, $m_2$, and $\theta$, defined in Section \ref{sec:membership}, associated with $M$. Let $\alpha=(\ceilings{\log{m_2}}+1)\theta$.
Assume that $\pair{M,x}$ has the form
$X\# S \# G \# P \#$ with $X$, $S$, $G$, and $P$ given in Section \ref{sec:membership}.

Into our storage alphabet, we include $k-1$ different symbols $B_1,B_2,\ldots,B_{k-1}$ satisfying $dv(B_i)=i$ for any $i\in[k-1]$ to express ``intermediate'' blank and we reserve $B$ for the frozen blank symbol.
For the ease of the description of later simulation, whenever we read $B_i$ for each $i\in[k-1]$, we automatically replace it with $B_{i+1}$, where $B_k$ is understood as $B$.

To represent the content of the depth-$k$ storage tape of $M$, we use a ``block'' of two ``sections''. Each section is a string of the form  $(\sigma\#)^{\theta}$ or $B_i^{\alpha}$ with $\sigma\in\tilde{\Gamma}$ and $i\in[k]$, and an entire block has the form $(\sigma\#)^{\theta}B_{2i+1}^{\alpha}$, $B_{2j+2}^{\alpha}(\sigma\#)^{\theta}$, or $B^{2\alpha}$ with $i<\frac{k-1}{2}$ and $j<\frac{k-2}{2}$.
Notice that each block has length exactly $2\alpha$ since $|\sigma\#|=\ceilings{\log{m_2}}+1$.
The string  $(\sigma\#)^{\theta}$ is called the \emph{core} of this block if any, and $\sigma$ is called the \emph{core value}.
The \emph{depth} of a block is the depth of its core.
For simplicity, we say that a block is \emph{currently accessed} if the storage-tape head is stationed in this block.
Similarly, $\sigma_h$, $p_h$, and $\gamma_h$ are said to be  \emph{currently accessed} if input-tape heads are reading them.
For technical reason, we further allow a block to
be $\Box^{2\alpha}$ and $\triangleright$. The tape content can be expressed as a series $Y_0 B Y_1 B \cdots B Y_{t}$ of blocks separated by $B$ with $Y_0 =  \: {\triangleright}$, where cell $t$ is the leftmost initially-blank cell. We denote this series by $Y$.
This is in part illustrated in Figure~\ref{fig:simulation-sda}.
We next decode this sequence $Y$ into a storage configuration $y_0y_1\cdots y_{t}$ as follows.
If $Y_l$ is $\triangleright$ and $\Box^{2\alpha}$, then we simply set $y_l=\, \triangleright$ and $y_l=\Box$, respectively.
If $Y_l$ is frozen blank (i.e., $Y_l=B^{2\alpha}$), then we set $y_l=B$.
Assume that $Y_l\notin\{\Box,\triangleright\}$ and the core value of $Y_l$ is $\sigma$. If we currently access $Y_l$ and $p_h$, then we set $y_l=(1^h,\sigma)$, which indicates that $M$ is in inner state $p_h$ and $\sigma$ appears in the cell of $M$'s storage-tape head location $l$.  In contrast, if $Y_l$ equals  $(\sigma\#)^{\theta}$ but $Y_l$ is not currently accessed, then we set $y_l=\sigma$.

In what follows, we want to demonstrate by induction on the number $r$ of steps that $\pair{M,x}$ correctly ``generates'' the $r$th storage  configuration of $M$ on $x$.
To express $Y$ at step $r$, we use the notation of $Y^{(r)}$ for clarity.  Similarly, we write  $y^{(r)}_l$ for $y_l$ at step $r$.
Initially, at time $t=1$, we produce $Y^{(0)} = Y_0^{(0)} B Y_1^{(0)}$ on the storage tape, where $Y^{(0)}_0 = \: {\triangleright}$ and $Y^{(0)}_1 =\Box^{2\alpha}$.
Our goal is to verify that the following statement (*) is true.

\begin{quote}
(*) For each index $r\geq1$, if $M$ scans $x_s$ and $\gamma$ in inner state $p$ and makes a transition at time $r$, then we read $X_s$ in $X$, $p$ in $S$, and $\gamma$ in $G$, find the corresponding encoded transition, simulate the $r$th step of $M$ by executing this encoded transition, and produce $Y^{(r)} = Y^{(r)}_0 B Y^{(r)}_1 B \cdots Y^{(r)}_{t}$ on the depth-$k$ storage tape, where  cell $t$ is the leftmost initially-blank cell. The sequence $y = y_0^{(r)}y_1^{(r)}\cdots y_t^{(r)}$ decoded from $Y^{(r)}$ then matches the storage configuration of $M$ on $x$ at step $r$.
\end{quote}

Meanwhile, we assume that (*) is true. After processing  $X_{n+1}\#$, since the storage configuration of $M$ must contain a halting state, we check whether this halting state is indeed $p_2$, which indicates that $M$ is in an accepting state. If so, we accept the input $\pair{M,x}$; otherwise, we reject it. Hereafter, we focus on proving (*) by induction on $r\in\nat^{+}$.

We wish to construct the desired $k$-sda$_2$($4$) universal simulator, say, $U$. For convenience, we call four input-tape heads of $U$ by heads 1--4,  where head 1 reads encoded transitions in the list $P$, head 2 ``points'' the encoding of the current inner state $p$ in the list $S$, head 3 accesses $X_s$ to retrieve $x_s$, and head 4 points the encoding of the current storage symbol $\gamma$ in $G$. During the following simulation, we remember the depth of the neighboring block from which the storage-tape head comes to the currently accessed block.

For readability, we purposely allow $U$ to change any symbol not in $\Gamma^{(k)}$ into $B$ in a single step.

\s

(Step $r=1$) Recall that $X_0 =\, {\triangleright}\,\#$.  At step $1$, since $M$ applies a transition of the form $\delta(p_1,\triangleright,\triangleright) = (p_{h_4},\triangleright,d_1,+1)$ given in the item (5) in Section \ref{sec:membership}, its storage-tape head moves from cell $0$ to cell $1$, and thus the storage configuration of $M$ on $x$ at step $1$ is $\triangleright\, (1^{h_4},\Box)$.
In scanning $\triangleright$ on the storage tape and $p_1$ in the list $S$, we search $P$ for
$\tilde{u}^{(+)}_{h',2}$ ($= \:{\triangleright} \# \:{\triangleright} \# p_{1} \# 001 \# \triangleright\! \# p_{h_4} \# 1^{d_1+1}0^{1-d_1} \#$)
with $h'=(1,h_4)$.
We move the storage-tape head to the right to access the first encountered $\Box$.
From the current tape content, it follows that  $Y^{(1)}_{0} = \triangleright$, $Y^{(1)}_1 = \Box^{2\alpha}$, and $Y^{(1)}_l = Y^{(0)}_{l}$ for any other indices $l$.
Thus, we obtain $y^{(1)}_0=\triangleright$ and $y^{(1)}_1= (1^{h_4},\Box)$.  Clearly, the resulting sequence $y^{(1)} =  y^{(1)}_0y^{(1)}_1$ matches the storage configuration of $M$ on $x$ at step $1$.

\s

(Step $r\geq2$) Let us consider the $r$th step of $M$.
Assume by induction  hypothesis that (*) is true for $Y^{(r-1)}$
at step $r-1$.
Take an arbitrary number $l\in[0,n+1]_{\integer}$
and assume that $Y^{(r-1)}_l$ is
of the form  $(\gamma_{h_1}\#)^{\theta}B_{2i+1}^{\alpha}$, $B_{2j+2}^{\alpha}(\gamma_{h_1}\#)^{\theta}$, or $B^{2\alpha}$ with $i<\frac{k-1}{2}$ and $j<\frac{k-2}{2}$.
If $Y_l^{(r-1)}$ is not currently accessed by $U$, then $Y_l^{(r)}$ equals $Y_l^{(r-1)}$.
In the following argument, we therefore assume otherwise.
Assume also that head 2 resides over the string $p_{h_2}$ in $S$, indicating that $M$ is in inner state $p_{h_2}$. Assume further that head 3 is reading $X_s$ that contains $x_s$, where $s\in[0,n+1]_{\integer}$. This means that  $M$'s tape head rests on cell $s$, which holds the input symbol $x_s$. At step $r$, $M$ makes one of the transitions (1)--(7) given in Section \ref{sec:membership}.
In what follows, we separately deal with those transitions.


(1)
Assume that $M$ applies a transition of the form
$\delta(p_{h_2},\sigma_{h_5},\gamma_{h_1}) = (p_{h_4},\gamma_{h_3},d_1,+1)$ with $\sigma_{h_5} = x_s$, $\gamma_{h_1}\in \Gamma^{(e)}$,    and $\gamma_{h_3}\in \Gamma^{(\min\{e+1,k\})}$ for an even number $e<k$.
In the currently accessed block, say, $Y^{(r-1)}_l$,
the storage-tape head is scanning $B$, $\Box$, or $\triangleright$.
We then discuss two cases (i) and (ii) separately.

(i) In the first case of $\min\{e+1,k\}<k$, we need to cope with the following subcases (a) and (b).

(a) Assuming that $\gamma_{h_1}\notin \{\triangleright,\Box\}$, we first find the location of the string $\gamma_{h_1}$ in $G$.
Since the currently accessed block $Y^{(r-1)}_l$ is of the form $(\gamma_{h_1}\#)^{\theta} B_e^{\alpha}$, we
move the storage-tape head rightward
through the first section $(\gamma_{h_1}\#)^{\theta}$.
We move head 4 through $G$ by picking substrings $\gamma\#$ one by one.
As we read one string of the form $\gamma_{h_1}\#$ in the section, we check whether $\gamma$ matches $\gamma_{h_1}$ by simultaneously comparing between $\gamma\#$ and $\gamma_{h_1}\#$ symbol by symbol. After each comparison, we mark $\gamma_{h_1}\#$ as ``being read'' by changing it into $B_{e+1}^{\ceilings{\log{c}}+1}$.
Even in the case of discovering any mismatch in the middle of  $\gamma_{h_1}\#$, we continue modifying the content of tape cells until we completely change this entire string $\gamma_{h_1}\#$ to  $B_{e+1}^{\ceilings{\log{c}}+1}$, and we then pick another substring and repeat the above procedure. This search is possible because we have an enough number of copies of the string $\gamma_{h_1}\#$ in $Y^{(r-1)}_l$, and thus we eventually find $\gamma_{h_1}$ in $G$.

Since heads 2--4 correctly point the locations of $x_s$, $p_{h_2}$, and $\gamma_{h_1}$, we can look for the correct encoded transition $\tilde{u}^{(+)}_{h,0}$ of the form $x_s\# \gamma_{h_{1}}\# p_{h_2} \# 001 \#      (\gamma_{h_3}\#)^{\theta}\# p_{h_4} \# 1^{d_1+1}0^{1-d_1}\#$ in $P$ by conducting the following search for
the receptor $x_s\# \gamma_{h_1} \# p_{h_2}\#$ of $\tilde{u}^{(+)}_{h,0}$.
In search of $\tilde{u}^{(+)}_{h,0}$, we first move head 1 to the first symbol of $P$
and, by moving head 1 through $P$, we pick encoded transitions, say, $\tilde{u}$ one by one.
We then examine that $(x_s,p_{h_2},\gamma_{h_1})$ truly appears in the receptor of $\tilde{u}$. If not, we continue picking another transition. This recursive process does not require any movement of $U$'s  storage-tape head.

Once we find $\tilde{u}^{(+)}_{h,0}$, by scanning its residue $001\# (\gamma_{h_3}\#)^{\theta}\# p_{h_4}\# 1^{d_1+1}0^{1-d_1}\#$, we overwrite the second section $B_e^{\alpha}$ of $Y_l^{(r-1)}$ by $(\gamma_{h_3}\#)^{\theta}$ symbol by symbol until reaching the frozen blank symbol $B$, which acts as a \emph{block separator}.
We then obtain the newly modified block $Y^{(r)}_{l}$, which has the form
$B_{e+1}^{\alpha} (\gamma_{h_3}\#)^{\theta}$.
Note that all the other blocks are intact.
We further move the storage-tape head rightward to the first non-$B$ symbol.
See Figure~\ref{fig:simulation-sda}(1)--(2) for an illustration
of storage-tape head moves.
Finally, we update the locations of heads 2 and 3 by moving head 2 to $p_{h_4}$, head 3 to $X_{s+d_1}$,
and head 4 to the first symbol of $P$.

From all the blocks at time $r$, we obtain $y_l^{(r)} = \gamma_{h_2}$,  $y^{(r)}_{l+1} = (1^{h_4},\xi)$ if $y^{(r-1)}_{l+1}=\xi$, and $y^{(r)}_j = y^{(r-1)}_j$ for all other indices $j$. By the definition, the resulting string $y^{(r)} = y_1^{(r)}y_2^{(r)}\cdots y_t^{(r)}$ matches the storage configuration of $M$ on $x$ at step $r$.

(b) Consider the next case of $\gamma_{h_1}=\Box$. Since $\gamma_{h_1}=\Box$,
we can search for $\gamma_{h_1}$ in $G$ by making only  storage-stationary moves. We then conduct the aforementioned receptor search to locate $\tilde{u}^{(+)}_{h,0}$ in $P$.
Since the currently accessed block $Y_l^{(r-1)}$ is $\Box^{2\alpha}$, we write  $B_1^{\alpha} (\gamma_{h_3}\#)^{\theta}$ into this block and move the tape head to the first encountered $\Box$.
We also move head 2 to $p_{h_4}$ and head 3 to $X_{s+d_1}$.

(ii) Consider the second case of $\min\{e+1,k\}=k$. Notice that $\gamma_{h_1}\neq \Box$. In this case, we write $B^{\alpha}$ instead of $(\gamma_{h_3}\#)^{\theta}$ in (i)(a), and $Y^{(r)}_l$ thus becomes $B^{2\alpha}$.

(2) Consider the case where $M$'s transition is of the form  $\delta(p_{h_2},\sigma_{h_5},\gamma_{h_1}) = (p_{h_4},\gamma_{h_3},d_1,-1)$ with $\sigma_{h_5}= x_s$, $\gamma_{h_1}\in \Gamma^{(e)}$, and $\gamma_{h_3}\in \Gamma^{(\min\{e+2,k\})}$ for an even number $e<k$. Notice that $\tilde{x}_s\neq \triangleright$.

(i)
Assume that  $\min\{e+2,k\}<k$ with $\sigma_{h_5}= x_s$.
Notice that the currently accessed block $Y^{(r-1)}_l$ is of the form $(\gamma_{h_1}\#)^{\theta} B_e^{\alpha}$. Firstly, we try to find $\gamma_{h_1}$ in $G$ by simultaneously moving head 4 and the storage-tape head from left to right over the first section $(\gamma_{h_1}\#)^{\theta}$ of  $Y^{(r-1)}_l$. After this search, this section $(\gamma_{h_1}\#)^{\theta}$  becomes $B_{e+1}^{\alpha}$.
We then search $P$ for
$\tilde{u}^{(lt)}_{h,0}$ ($\equiv  \sigma_{h_5}\# \gamma_{h_{1}} \# p_{h_2} \# 100 \# (\gamma_{h_3}\#)^{\theta} \# p_{h_4} \# 1^{d_1+1}0^{1-d_1}\#$) by the aforementioned receptor search. Once we find $\tilde{u}^{(lt)}_{h,0}$, we locate the residue $100 \# (\gamma_{h_3}\#)^{\theta} \# p_{h_4} \# 1^{d_1+1}0^{1-d_1}\#$.

Next, we overwrite the second section of $Y_l^{(r-1)}$ by $\hat{B}^{\alpha}$, where $\hat{B}$ is a fresh storage symbol.
When we reach a block separator, we make a left turn and then change $\hat{B}^{\alpha}$ to $B_{e+2}^{\alpha}$ from right to left. We continue overwriting the next section $B_{e+1}^{\alpha}$ by $(\gamma_{h_3}\#)^{\theta}$ until reaching a block separator. In the end, we step to the left.
Since $Y_{l-1}^{(r)} = Y_{l-1}^{(r-1)}$ and $Y^{(r)}_l =  (\gamma_{h_3}\#)^{\theta} B_{e+2}^{\alpha}$, we obtain $y^{(r)}_{l} = \gamma_{h_3}$ and $y^{(r)}_{l-1} = (1^{h_4},\xi)$ if $y^{(r-1)}_{l-1}=\xi$. All the other $y^{(r-1)}_j$ are updated to be $y^{(r)}_j$. The resulting sequence $y^{(r)}$ clearly matches the storage configuration of $M$ on $x$ at step $r$. Figure~\ref{fig:simulation-sda}(3)--(4) illustrate this process.

(ii) In the case of $\min\{e+2,k\}=k$, since $\gamma_{h_3}$ is $B$, we  write $B^{\alpha}$ instead of $(\gamma_{h_3}\#)^{\theta}$ in (i). As a result,   $Y^{(r)}_{l}$ becomes $B^{2\alpha}$.

(3) This case is symmetric to (1). Since $M$'s storage-tape head earlier came from the right, $U$'s storage-tape head starts at the rightmost symbol of the second section of $Y_l^{(r-1)}$. During the search for $\gamma_{h_1}$ in $G$, we need to move the storage-tape head leftward. After finding $\tilde{u}^{(-)}_{h,0}$ by the receptor search, we overwrite the first section of $Y_l^{(r-1)}$ by $(\gamma_{h_3}\#)^{\theta}$ from right to left.

(4) This case is handled symmetrically to (2) in a way similar to (3).

(5) In the case of $\triangleright$, we look for $\tilde{u}^{(+)}_{h',2}$ by conducting the aforementioned search for its receptor $x_s\#\,  \triangleright\,\#p_{h_2}\#$. Since $dv(\triangleright)=k$, we do not need to overwrite $\triangleright$ on the storage tape.

(6) Consider the case where $M$'s transition is of the form $\delta(p_{h_2},\sigma_{h_5},\gamma_{h_1})= (p_{h_4},B,0,d_2)$ with $e\in\{k-1,k\}$ and  $d_2\in\{-1,+1\}$. We discuss two subcases of $e=k-1$ and $e=k$ separately.

(i) Assume that $e=k-1$. In the case where $e$ is even, $M$'s storage-tape head earlier came from the left. Note that the other case where $e$ is odd is symmetric.
By scanning the first section of $Y^{(r-1)}_l$, we look for $\gamma_{h_1}\#$ in $G$ as
we overwrite this section by $B^{\alpha}$ until we encounter the second section $B_{k-1}^{\alpha}$.
Once we find $\gamma_{h_1}\#$ in $G$, we conduct the receptor search for $\tilde{u}^{(d_2)}_{h',2}$ and then retrieve its residue  to discover the value of $d_2$.

Since we remember the depth of the left adjacent block, let $g$ denote this depth.
Consider the first case of $g<k$.
We change the second section $B_{k-1}^{\alpha}$ by $B^{\alpha}$ and reach a  block separator.
When $d_2=-1$, we make a left turn and move leftward to the first non-$B$  symbol. In the case of $g=k$, on the contrary,
if $d_2=+1$, then we continue overwriting the section $B_{k-1}^{\alpha}$ by $B^{\alpha}$ until reaching a block separator. We then step to the next block.
Finally, when $d_2=-1$, based on the fact that the left adjacent cell is frozen blank, without moving the storage-tape head, we can determine whether, at the next step $r+1$, $M$ makes a right turn or moves further to the left.

(a) Assume that $M$ make a right turn at step $r+1$. We first overwrite $B_{k-1}^{\alpha}$ by $B^{\alpha}$ until reaching a block separator. We update the location of head 2 from $p_{h_2}$ to $p_{h_4}$, update the depth of the neighboring block to $k$, and move to the first symbol of the next block.

(b) If $M$ moves further to the left at step $r+1$, then we overwrite $B_{k-1}^{\alpha}$ by $B^{\alpha}$ until reaching a block separator. We then make a left turn, move head 2 to $p_{h_4}$ from $p_{h_2}$, and move the storage-tape head leftward to the first non-$B$ symbol. Since $M$ is depth-susceptible, $M$ does not change its behavior while reading $B$s, and thus we can smoothly pass through all blank blocks without any further information.

(ii) On the contrary, if $e=k$, then $h_2=h_4$ follows. If $g=k$, then $d_2$ must be $+1$. Since $M$ does not change its inner state and head direction, we move rightward until encountering the first non-$B$ symbol.
In the case of $g<k$, it is possible to look for $\tilde{u}^{(d_2)}_{h',2}$ by conducting the receptor search and discover $d_2$. If $d_2=+1$, then we move rightward to the first non-$B$ symbol. In contrast, when $d_2=-1$, we skip the scanning of $Y^{(r-1)}_l$ and make a left turn to the first non-$B$ symbol.

(7) Assume that $M$ makes a storage-stationary move with $\gamma_{h_1}\in\Gamma^{(e)}$. We first look for
$\gamma_{h_1}$ in $G$ by reading $(\gamma_{h_1}\#)^{\theta}$ in $Y^{(r-1)}_l$ either from left to right or from right to left, depending on the value of $e$. After reading it, it is overwritten by $B_{e+1}^{\alpha}$.
We then search for the receptor of the form $x_s\# \gamma_{h_1} \# p_{h_2}\#$ to locate $\tilde{u}^{(0)}_{h',3}$ in $P$.
After finding $\tilde{u}^{(0)}_{h',3}$, we retrieve the residue of the form $000\# (\gamma_{h_1}\#)^{\theta} \# p_{h_4}\# 1^{d_1+1}0^{1-d_1}\#$.
Since $M$'s storage-tape head does not move, we stop $U$'s storage-tape head in the middle of the current block just after finishing the replacement of $(\gamma_{h_1}\#)^{\theta}$ by $B_{e+1}^{\alpha}$.
Note that, if $M$'s input-tape head moves, then $\gamma_{h_1}\notin \Gamma^{(k-1)}\cup \Gamma^{(k)}$ follows because, otherwise, we obtain $d_1=0$, a contradiction. When $M$'s input-tape head does not move, nonetheless, $\gamma_{h_1}\neq B$ follows because, otherwise, $M$ cannot halt.
Note that, as long as $M$ makes storage-stationary moves, we reuse the obtained information on $\gamma_{h_1}$ without further moving $U$'s storage-tape head. When $M$ eventually moves its storage-tape head either to the right or to the left, $M$ must take one of the transitions (1)--(6). We then simulate the transition of $M$ in a way explained in (1)--(6) except for the first process of finding $\gamma_{h_1}$ in $G$ since this process has already been done.

\s

Finally, we note that the constructed universal simulator $U$ is indeed a depth-immune $k$-sda$_{2}$($3$). This completes the proof.
\end{proofof}

To close Section \ref{sec:hardest-language}, we still need to prove Theorem \ref{hardness-prop}.

\begin{proofof}{Theorem \ref{hardness-prop}}
(1) This directly comes from Theorem \ref{auxiliary-characterize} and the depth-immune $k$-sda$_{2}$(4) universal simulator for all depth-susceptible $k$-sda's, guaranteed by Lemma \ref{universal-simulator}.

(2) Let $M=(Q,\Sigma,\{\Gamma^{(e)}\}_{e\in[0,k]_{\integer}}, \{\triangleright, \triangleleft \}, \delta,q_0,Q_{acc},Q_{rej})$ be any depth-susceptible $k$-sda.
We want to show that $L(M)$ is $\dl$-m-reducible to $\mathrm{MEMB}_k$.
Take any input $x$ and consider the encoding $\pair{M,x}$ of $M$ and $x$, where $\pair{M,x}$ has the form
$X\# S\# G\# P\#$ as defined in Section \ref{sec:membership}.
Since $Q$, $\Sigma$, and $\Gamma$ are all fixed finite sets independent of $x$, the construction of $\pair{M,x}$ from $x$ requires polynomial time using only logarithmic space. We define $f_M(x) =\pair{M,x}$ for any $x\in\Sigma^*$.
By the definition of $\mathrm{MEMB}_k$, it follows that $M$ accepts $x$ iff $f_M(x)$ belongs to $\mathrm{MEMB}_k$. Therefore, we conclude that $L(M)$ is $\dl$-m-reducible to $\mathrm{MEMB}_k$ via $f_M$.
Since $\mathrm{LOG}k\mathrm{SDA}$ is closed under $\dl$-m-reductions, $\mathrm{MEMB}_k$ is also $\dl$-m-hard for $\mathrm{LOG}k\mathrm{SDA}$.
\end{proofof}


\section{A Complexity Upper Bound of $k$SDA$_{imm}$}\label{sec:upper-bounds}

We have discussed the mathematical model of depth-susceptible $k$-sda's in Sections \ref{sec:auxiliary}--\ref{sec:hardest-language}. We here turn our interest to another model, depth-immune $k$-sda's, and discuss the computational complexity of $k\mathrm{SDA}_{imm}$. Concerning the complexity of $\dcfl$, Cook \cite{Coo79} earlier  demonstrated that $\dcfl$ is included in $\mathrm{SC}^2$.
In what follows, for any $k\geq2$, we intend to present a non-trivial $\mathrm{SC}^{k}$-upper bound on the complexity of $k\mathrm{SDA}_{imm}$, namely, $k\mathrm{SDA}_{imm} \subseteq \mathrm{SC}^k$. Since $k\mathrm{SDA}_{imm}$ properly includes $\dcfl$, this upper bound of $k\mathrm{SDA}_{imm}$ significantly  extends Cook's result in \cite{Coo79}.

\begin{theorem}\label{upper-bound}
For any integer $k\geq2$, $k\mathrm{SDA}_{imm} \subseteq \mathrm{SC}^k$. Thus, $\omega\mathrm{SDA}_{imm} \subseteq \mathrm{SC}$ follows.
\end{theorem}

Since $\mathrm{SC}^k$ is closed under $\dl$-m-reductions, Theorem \ref{upper-bound} instantly yields the following corollary.

\begin{corollary}
For any $k\geq2$, $\mathrm{LOG}k\mathrm{SDA}_{imm} \subseteq \mathrm{SC}^k$.
\end{corollary}

Another immediate consequence of Theorem \ref{upper-bound} is a complexity upper bound of Hibbard's \emph{deterministic $k$-limited automata} ($k$-lda's, for short) because $k$-lda's (with the blank-skipping property \cite{Yam19}) can be easily simulated by depth-immune $k$-sda's by pretending that all input symbols are written on a storage tape.
Notice that, in the past literature, no upper bound except for $\cfl$ has been shown for Hibbard's language families \cite{Hib67}.
Thus, this is the first time to show non-trivial upper bounds for those families.

\begin{corollary}
For any $k\geq 2$, all languages recognized by Hibbard's $k$-lda's are in $\mathrm{SC}^k$.
\end{corollary}


To verify Theorem \ref{upper-bound}, we attempt to employ a \emph{divide-and-conquer argument} to simulate the behaviors of each depth-immune $k$-sda in polynomial time using only $O(\log^k{n})$ space. Our simulation procedure is  based on \cite{BCMV83} but expanded significantly to cope with more complex moves of depth-immune $k$-sda's.

\subsection{Markers and Contingency Trees}\label{sec:contingency-tree}

We first need to lay out a basic framework to describe the desired  procedures simulating any depth-immune $k$-sda's.

Let $M=(Q,\Sigma,\{\Gamma^{(e)}\}_{e\in[0,k]_{\integer}}, \{{\triangleright,\triangleleft}\}, \delta, q_0,Q_{acc},Q_{rej})$ denote any depth-immune $k$-sda.
We fix an arbitrary input $x\in\Sigma^*$ and set $n=|x|$. Since $M$ halts in polynomial time, we choose an appropriate polynomial $p$ so that, for any string $x$, $p(|x|)$ upper-bounds the running time of $M$ on input $x$.
To simulate the behavior of $M$ on $x$, we introduce an important notion of ``marker''.
Any \emph{computation} of a depth-immune $k$-sda is characterized by a series of such markers.
A \emph{marker} $C$ is formally a quintuple $(q,l_1,l_2,\sigma,r,t)$ that indicates the following circumstances: at time $t$ with section time $r$ (which will be explained later), $M$ is in inner state $q$, its input-tape head is located at $l_1$, and the storage-tape head is at the $l_2$th cell, which contains symbol $\sigma$. To express the entries of $C$, we use the following specific notations: $state(C)=q$, $in\mbox{-}loc(C)=l_1$, $st\mbox{-}loc(C)=l_2$, $symb(C)=\sigma$,  $sectime(C)=r$, and $time(C)=t$.
Let $\MM_x$ denote the set of all markers of $M$ on $x$.
To emphasize ``$l_2$'', in particular, we often call $C$ an \emph{$l_2$-marker} if $st\mbox{-}loc(C)=l_2$. Similarly, we call $C$ a \emph{$t$-marker} if $time(C)=t$.

Assume that $C_t$ is a market of the form $(q,l_1,l_2,\sigma,r,t)$ at time $t$ and $M$ executes a transition of the from $\delta(q,x_{(l_1)},\sigma)=(p,\tau,d_1,d_2)$ with $d_1\in\{0,+1\}$ and $d_2\in\{-1,0,+1\}$. We then set $next\mbox{-}state(C_t)=p$, $next\mbox{-}loc(C_t)=l_2+d_2$,
and $next\mbox{-}symb(C_t)=\tau$.
To obtain the next marker $C_{t+1}$ (at time $t+1$) from $C_t$, we need to know the current content of cell $l_2+d_2$, say, $\xi$.
To obtain this symbol $\xi$, we first compute
the most recent  $(l_2+d_2)$-marker $C'$ with $time(C')<t$ and then set $\xi= next\mbox{-}symb(C')$.  The quintuple $(p,l_1+d_1,l_2+d_2,\xi,r+|d_2|,t+1)$ thus becomes the desired marker $C_{t+1}$.


\begin{figure}[t]
\centering
\includegraphics*[height=5.3cm]{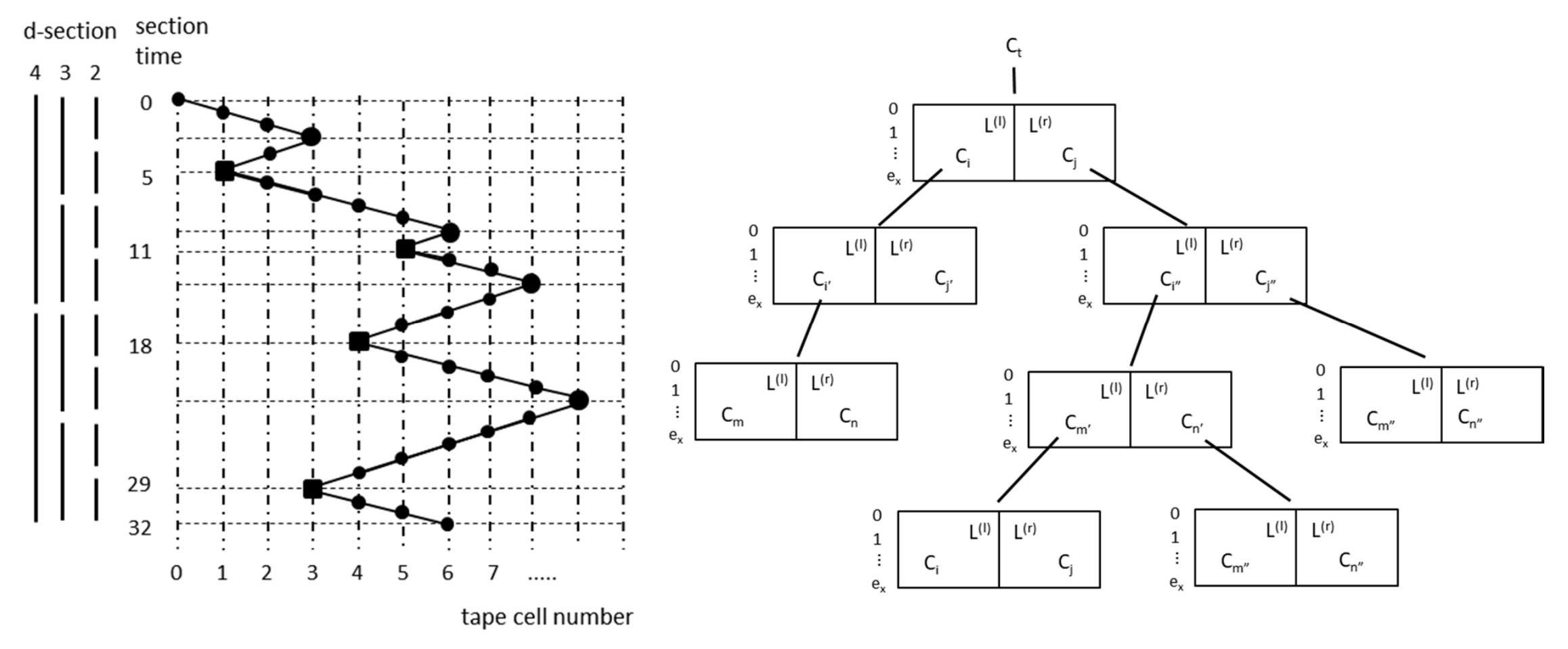}
\caption{[left] A history of consecutive moves of a storage-tape head for $k=3$. The leftmost three vertical dashed lines indicate $d$-sections for $d=3,4,5$. The other vertical dashed lines indicate the tape cell numbers from $0$ to $10$ and the horizontal dotted lines show section time from $0$ to $32$.  All storage-stationary-moves are suppressed into black circles and boxes.
[right] A contingency tree, in which each node (except for the root) is a contingency list linked to a marker in the parent node.}\label{fig:storage-tape-content}
\end{figure}


We then introduce additional critical notions. Let $C=(q,l_1,l_2,\sigma,r,t)$ and $C'=(q',l'_1,l'_2,\sigma',r',t')$ be two arbitrary markers of $M$ on $x$. We say that $C'$ is \emph{left-visible} from $C$ if
(i) $l'_2 < l_2$ and $t'<t$ and
(ii) there is no other marker $\tilde{C}$ satisfying both $st\mbox{-}loc(\tilde{C}) \leq   l'_2$ and $t'+1\leq time(\tilde{C})<t$.
Moreover, $C'$  is the \emph{left-cut} of $C$ if $C'$ is the most recent left-visible marker from $C$; that is, $C'$ is left-visible from $C$ and $t'$ is the largest number satisfying $t' < t$.
If the storage-tape head moves to the left from cell $l_2$ at time $t$, then the left cut of $C_t$ provides the latest content of cell $l_2-1$. This makes it possible to determine the renewed content of cell $l_2-1$ by applying $\delta$ directly. In symmetry, we can define the notions of \emph{right-visibility} and \emph{right-cut}.

A \emph{$0$-section} consists of all markers indicating either (a) a right/left turn (i.e., leftmost/rightmost point) or (b) a non storage-stationary move to the right/left followed by a (possibly empty) series of consecutive storage-stationary moves.
For any $d\in\nat$, a \emph{$(d+1)$-section} is the union of two consecutive $d$-sections. The \emph{section time} of a marker $C$ is the total number of $0$-sections before $C$ (not including the $0$-section containing $C$).

Given a set of markers, a marker $C$ is said to be the \emph{leftmost marker} (resp., the \emph{rightmost marker}) if $st\mbox{-}loc(C)$ is the smallest (resp., the largest) among the markers in the given set.
The leftmost (resp., rightmost) marker in each $d$-section $S$ is called the \emph{left-representative} (resp., \emph{right-representative}) of $S$.
Given a marker $C$, a section $S$ is called \emph{current} for $C$ if $S$ contains $C$, $S$ is \emph{completed} for $C$ if $C$ appears after the end of $S$ according to time, and $S$ is \emph{last-left-good} for $C$ if  $S$ is the latest completed section whose left-representative is left-visible from $C$. In a similar way, we define the notion of ``last-right-goodness''.


Next, we wish to introduce the notion of contingency tree.
Given any string $x$, we conveniently set $e_x=\ceilings{\log|x|}$.
Let $C_t=(q,l_1,l_2,\sigma,r,t)$ denote an arbitrary marker of $M$ on $x$.
To explain a contingency tree, we first define a \emph{contingency list at time $t$}, which consists of the sets described below. Let $d\in[0,e_x]_{\integer}$.

\ms
\n{\bf [Definition of Contingency Lists]}

\renewcommand{\labelitemi}{$\circ$}
\begin{enumerate}\vs{-2}
  \setlength{\topsep}{-2mm}%
  \setlength{\itemsep}{1mm}%
  \setlength{\parskip}{0cm}%

\item[(a)] $L_{last}^{(l)}(d,l_2,t)$, which consists of all markers $C$ satisfying the following requirement: there exist a $d$-section $S$ and a $(d+1)$-section $S'$ such that (i) $C$ is the left-representative of $S$, (iii) $S$ is enclosed in $S'$, and (iv) $S'$ is the last-left-good section for $C_t$.

\item[(b)] $L_{cur}^{(l)}(d,l_2,t)$, which consists of all markers $C$ satisfying the following requirement: there exist a $d$-section $S$ and a $(d+1)$-section $S'$ for which (i) $C$ is the left-representative of $S$, (ii) $S$ is enclosed in $S'$, (iii) $S'$ is current for $C_t$, and (iv) $C$ is left-visible from $C_t$.

\item[(c)] $L_{last}^{(r)}(d,l_2,t)$ and $L_{cur}^{(r)}(d,l_2,t)$, which are defined similarly to (a) and (b) by simply replacing ``left'' with ``right''.
\end{enumerate}


\begin{figure}[t]
\centering
\includegraphics*[height=5.0cm]{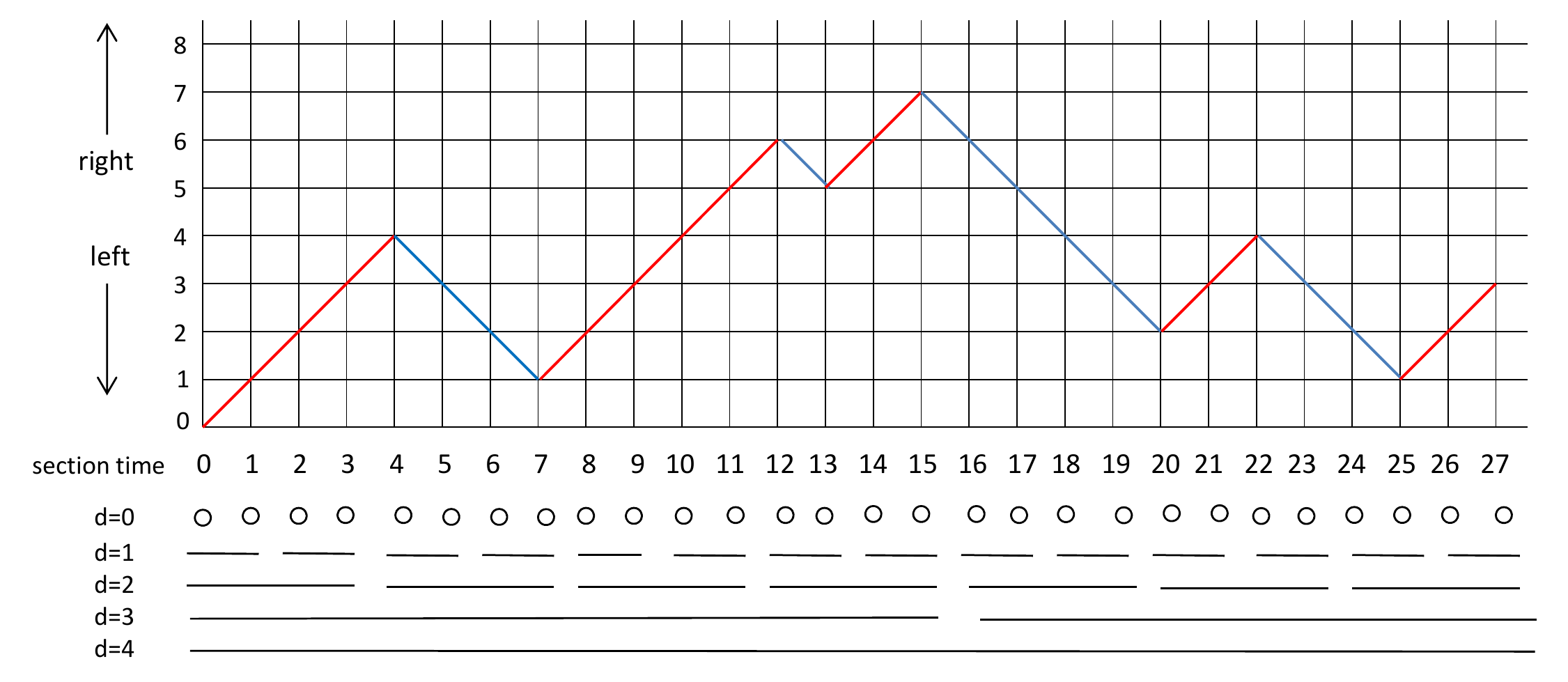}
\caption{A storage content history. Here, we assume that an underlying $k$-sda makes no storage-stationary-move. Thus, runtime matches section time.}\label{fig:storage-hisotry}
\end{figure}


Notice that, by the definition, the left-cut (resp., the right-cut) of  $C_t$ is the latest (i.e., the most recent) marker in $L^{(l)}(0,l_2,t)$ (resp., $L^{(r)}(0,l_2,t)$).
For every index $d\in[0,e_x]_{\integer}$ and any $a\in\{l,r\}$, let $L^{(a)}(d,l_2,t) = L^{(a)}_{last}(d,l_2,t)\cup L^{(a)}_{cur}(d,l_2,t)$ and set
$L(e_1,e_2,l_2,t) = (\bigcup_{d\in[0,e_1]_{\integer}} L^{(l)}(d,l_2,t)) \cup (\bigcup_{d\in[0,e_2]_{\integer}} L^{(r)}(d,l_2,t))$ for any pair $e_1,e_2\in[0,e_x]_{\integer}$.
We always assume that all markers in $L(e_1,e_2,l_2,t)$ are enumerated according to time.
This contingency list $L(e_1,e_2,l_2,t)$ is said to be \emph{linked to} $ C_t$.
A \emph{$2$-contingency tree at time $t$} is a two-node tree whose root is $C_t$ and its child is a contingency list linked to $C_t$.  A \emph{$(k+1)$-contingency tree at time $t$} is a leveled, rooted tree whose top 2 levels form a $2$-contingency tree at time $t$ and each marker (except for the root $C_t$) appearing in this $2$-contingency tree is the  root of another $k$-contingency tree at time $t'$ with $t'<t$.  Figure~\ref{fig:storage-tape-content} illustrates such a contingency tree. In particular, we treat $C_0$ as a special marker that links itself to a unique \emph{empty contingency list}. To describe a $k$-contingency tree at time $t$, we use the notation $\DD(k,e_1,e_2,l_2,t)$. The \emph{principal contingency list} of such a tree refers to the child of the root $C_t$ of the tree.

\begin{lemma}\label{marker-number}
For any pair $e_1,e_2\in[0,e_x]_{\integer}$, the
total number of markers stored inside $\DD(k,e_1,e_2,l_2,t)$ is $O(\log^{k-2}n)$, where $n=|x|$.
\end{lemma}

\begin{proof}
Let $n=|x|$. Each contingency list inside $\DD(k,e_1,e_2,l_2,t)$ has only $O(\log{n})$ markers because $0\leq e_1,e_2\leq e_x=\ceilings{\log{n}}$.
Each node  in $\DD(k,e_1,e_2,l_2,t)$ except for the root is a contingency list and there are $k-1$ levels of such nodes in the tree. Since the second level of the tree has only one node, $\DD(k,e_1,e_2,l_2,t)$ must contain $O(\log^{k-2}n)$ markers.
\end{proof}

\subsection{Description of the Main Subroutine}\label{sec:description}

We will describe in Section \ref{sec:complexity-analysis} how to simulate a  polynomial-time depth-immune $k$-sda  using only $O(\log^k{n})$ space.
The core of this simulation procedure is the following subroutine, Subroutine $\AAA$, which manages  $k$-contingency trees to recover the past behaviors of a depth-immune $k$-sda for determining its next moves.

\ms
\n{\sc Subroutine $\AAA$:}
\renewcommand{\labelitemi}{$\circ$}
\begin{itemize}\vs{-2}
  \setlength{\topsep}{-2mm}%
  \setlength{\itemsep}{1mm}%
  \setlength{\parskip}{0cm}%

\item[] An input is of the form $(C_t,dv,e_1,e_2,,\DD(dv,e_1,e_2,l_2,t))$ with $C_t=(q,l_1,l_2,\sigma,r,t)\in \MM_x$, $dv\in[0,k]_{\integer}$, and $e_1,e_2\in[0,e_x]_{\integer}$.  If $dv\geq dv(C_{t+1})$ holds for the $(t+1)$-marker $C_{t+1} = (q',l'_1,l'_2,\sigma',r',t+1)$, then $\AAA$ returns  $(C_{t+1},dv,e_1,e_2,\DD(dv,e_1,e_2,l'_2,t+1))$  by making a (possible) series of recursive calls to itself. Otherwise, it may return anything.
\end{itemize}


Hereafter, we explain how to implement Subroutine $\AAA$. We assume that all markers in each contingency list are automatically enumerated according to time.


\begin{figure}[t]
\centering
\includegraphics*[height=5.6cm]{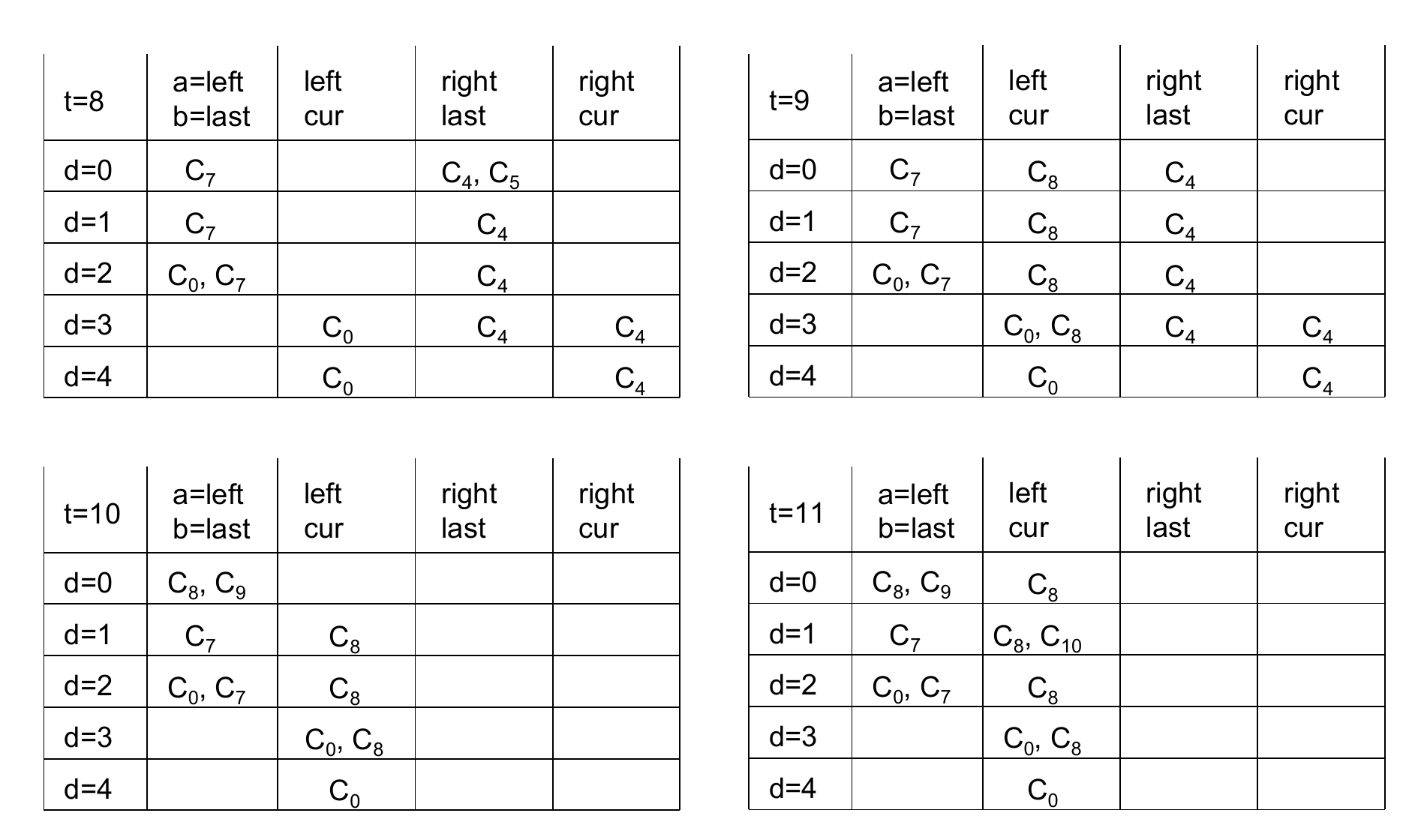}
\caption{The content of the principal contingency list at time $t=8,9,10,11$. The notation $C_i$ denotes a marker at (section) time $i$ given by Figure~\ref{fig:storage-hisotry}. The parameters $a$ and $b$ satisfy  $a\in\{l(eft),r(ight)\}$ and $b\in\{last,cur\}$ in $L^{(a)}_{b}(d,l_2,t)$.}\label{fig:storage-tree-table-01}
\end{figure}


\begin{figure}[t]
\centering
\includegraphics*[height=5.6cm]{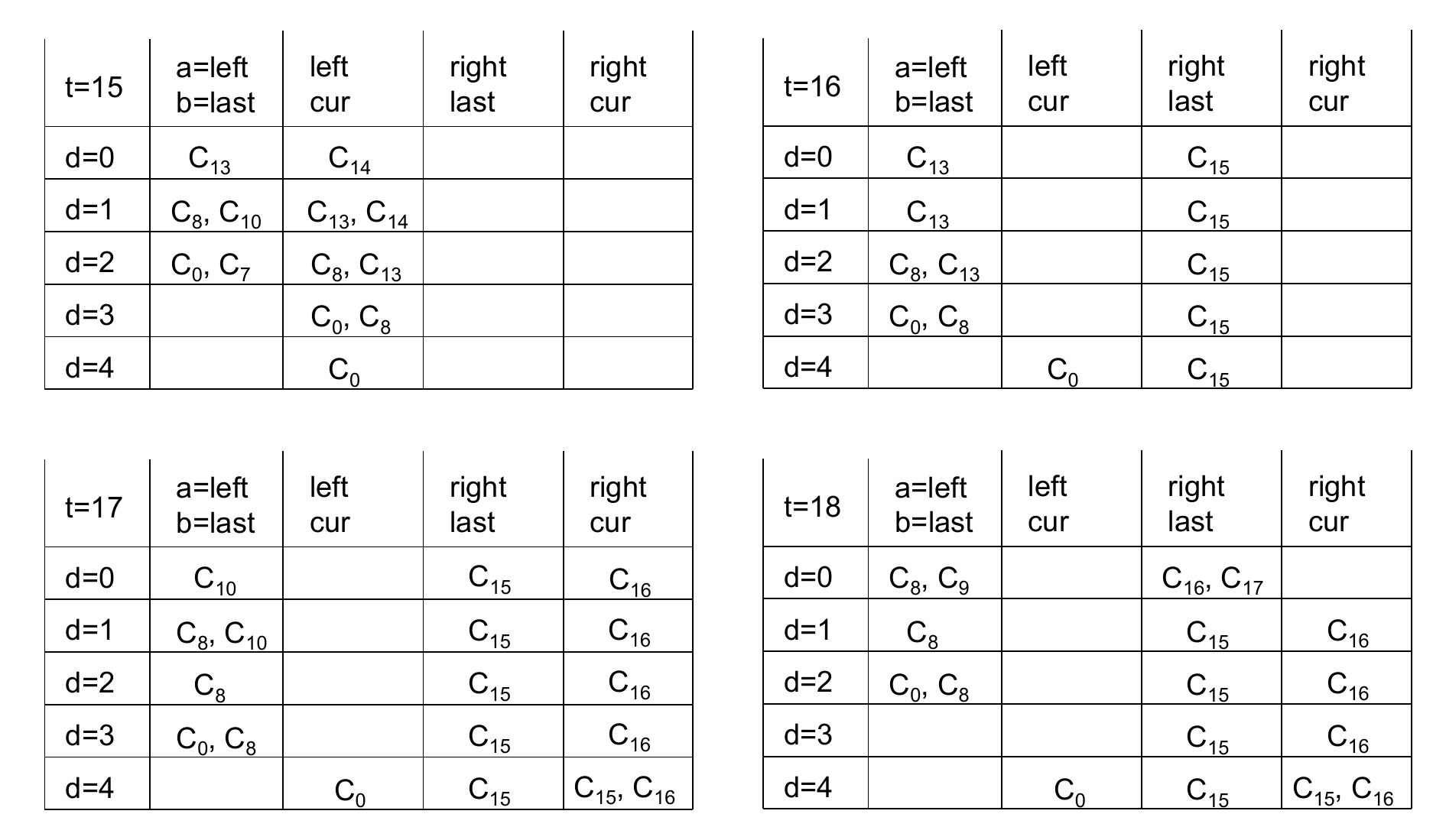}
\caption{The content of the principal contingency list at (section) time $t=15,16,17,18$. The parameters $a$ and $b$ are the same as in Figure~\ref{fig:storage-tree-table-01}.}\label{fig:storage-tree-table-02}
\end{figure}


\ms
\n{\bf [Implementation of Subroutine $\AAA$]}
\renewcommand{\labelitemi}{$\circ$}
\begin{enumerate}\vs{-1}
  \setlength{\topsep}{-2mm}%
  \setlength{\itemsep}{1mm}%
  \setlength{\parskip}{0cm}%

\item Consider the case where the storage-tape head makes a stationary move, that is, $l'_2=l_2$. Define $r'=r$.
    Since we do not need to alter the contingency tree,
    we automatically set $L^{(l)}(d,l'_2,t+1) = L^{(l)}(d,l_2,t)$ and $L^{(r)}(d',l'_2,t+1) = L^{(r)}(d',l_2,t)$ for any $d\in[0,e_1]_{\integer}$ and $d'\in[0,e_2]_{\integer}$ in the principal contingency list of the contingency tree  $\DD(dv,e_1,e_2,l'_2,t+1)$.
    We further compute $C_{t+1} = (q',l'_1,l'_2,\sigma',r',t+1)$ from $C_t$ directly by applying $\delta$.

\item Consider the case where the storage-tape head is moving to the right, that is, $l'_2=l_2+1$. Define $r'=r+1$ and compute $C_{t+1}$ and $\DD(dv,e_1,e_2,l'_2,t+1)$ as follows.

\begin{enumerate}
  \setlength{\topsep}{-2mm}%
  \setlength{\itemsep}{1mm}%
  \setlength{\parskip}{0cm}%

\item[(a)] [list $L^{(l)}(d,l'_2,t+1)$ for any $d\in[0,e_1]_{\integer}$] Consider the principal contingency list $L(e_1,e_2,l_2,t)$ of $\DD(dv,e_1,e_2,l_2,t)$. We compute $L(e_1,e_2,l'_2,t+1)$ in the following way.
    Take the maximum numbers $d_0\in [0,e_1]_{\integer}$ and $d_2\in[0,e_1]_{\integer}$ satisfying that $r'\equiv 0$ ($\mathrm{mod}\:2^{d_0}$) and $r\equiv0$ ($\mathrm{mod}\:2^{d_2}$). For each index $i\in[0,\max\{d_0,d_2\}-1]_{\integer}$, we inductively define $A_i$ by setting $A_0=L^{(l)}_{cur}(0,l_2,t)\cup\{C_t\}$ and $A_{i} = L^{(l)}_{cur}(i-1,l_2,t)\cup\{C'_{i-1}\}$, where $C'_{i-1}$ is the oldest marker in $A_{i-1}$.

\begin{enumerate}
  \setlength{\topsep}{-2mm}%
  \setlength{\itemsep}{1mm}%
  \setlength{\parskip}{0cm}%

\item[(i)] Assume that $r'$ is even. Note that $d_2=0$. For each $i\in[0,d_0-1]_{\integer}$, we set $L^{(l)}_{cur}(i,l'_2,t+1) = \setempty$, $L^{(l)}_{last}(i,l'_2,t+1) = A_i$, $L^{(l)}_{cur}(d_0,l'_2,t+1) = A_{d_0}$, and $L^{(l)}_{last}(d_0,l'_2,t+1) = L^{(l)}_{last}(d_0,l_2,t)$.
    For any $i\in[d_0+1,e_1]_{\integer}$ and any subscript $b\in\{cur,last\}$, we further define $L^{(l)}_b(i,l'_2,t+1) = L^{(l)}_b(i,l_2,t)$. In the example of $t=9$ in Figure~\ref{fig:storage-tree-table-01}, we have $r'=10$ and $d_0=1$.

\item[(ii)] If $r'$ is odd, then we set $L^{(l)}_{cur}(i,l'_2,t+1) = A_i$ for any $i\in[0,d_2]_{\integer}$.
    For all the other indices $d$ and all subscripts $b\in\{cur,last\}$, we set $L^{(l)}_{b}(d,l'_2,t+1) = L^{(l)}_{b}(d,l_2,t)$.
    We also update the others without changing the content. In the example of $t=8$ in Figure~\ref{fig:storage-tree-table-01}, we have $r'=9$, $d_0=0$, and $d_2=3$.

\item[(iii)] For all the remaining ``nodes'' in the contingency tree $\DD(dv,e_1,e_2,l_2,t)$, we also update them in accordance with the above changes.
\end{enumerate}

\item[(b)] [list $L^{(r)}(d',l'_2,t+1)$ for any  $d'\in[0,e_2]_{\integer}$]  This process is obtained straightforwardly from Item 3(a) by replacing ``$l$'' with ``$r$'', ``left'' with ``right'',  ``$d'$'' with ``$d$'', and ``$e_1$'' with ``$e_2$''.

\item[(c)] [marker $C_{t+1}$] To obtain the marker $C_{t+1} = (q',l'_1,l'_2,\sigma',r',t+1)$ from $C_t$, it suffices to compute the left-cut (if any).

\begin{enumerate}
  \setlength{\topsep}{-2mm}%
  \setlength{\itemsep}{1mm}%
  \setlength{\parskip}{0cm}%

\item[(i)] Assume that $dv\geq0$. If there exists a left-cut $C''$ of $C_t$ at time $<t$ either stored inside $\DD(dv,e_1,e_2,l_2,t)$ or obtained by Item 2(b), then the desired marker $C_{t+1}$ is directly calculated from $C_t$ and $symb(C'')$ by applying $\delta$.
    If there is no right-representative at location $\geq l'_2$ at time $<t$ inside $\DD(dv,e_1,e_2,l_2,t)$, then we know that the storage-tape head has never visited cell $l'_2$, and thus cell $l'_2$ must contain $\Box$. By setting  $\sigma'= \Box$, we compute the $l'_2$-marker $C_{t+1}$ from $C_t$ alone  by $\delta$.

\item[(ii)] If $dv<0$, then cell $l'_2$ must be already frozen. Let  $\sigma'= B$ and produce $C_{t+1}$ from $C_t$ alone by applying $\delta$.
\end{enumerate}
\end{enumerate}

\item Let us consider the case where the storage-tape head is moving to the left, that is, $l'_2=l_2 -1$. We set $r'=r+1$ and compute $C_{t+1}$ and   $\DD(dv,e_1,e_2,l'_2,t+1)$ in the following way.

\begin{enumerate}
  \setlength{\topsep}{-2mm}%
  \setlength{\itemsep}{1mm}%
  \setlength{\parskip}{0cm}%

\item[(a)] [list $L^{(l)}(d,l'_2,t+1)$ for any $d\in[0,e_1]_{\integer}$] Consider the principal contingency list $L(e_1,e_2,l_2,t)$ of $\DD(dv,e_1,e_2,l_2,t)$ and compute $L(e_1,e_2,l'_2,t+1)$ by taking the following procedure. Take the maximum number $d_0$ in $[0,e_1]_{\integer}$ satisfying $r'\equiv 0$ ($\mathrm{mod}\:2^{d_0}$) and define $d_1$ to be the minimal number in $[0,e_1]_{\integer}$ for which $|L^{(l)}(d_1,l_2,t)|\geq2$ holds. In case that no such $d_1$ exists, we automatically set $d_1=e_1$.

\begin{enumerate}
  \setlength{\topsep}{-2mm}%
  \setlength{\itemsep}{1mm}%
  \setlength{\parskip}{0cm}%

\item[(i)] For each index $d\in[0,d_1]_{\integer}$, we reset $L^{(l)}(d,l_2,t)$ to be $L^{(l)}(d,l_2,t)-\{C'\}$, where $C'$ is the newest marker in $L^{(l)}(d,l_2,t)$. The contingency tree $\DD(dv,e_1,e_2,l_2,t)$ is then updated.
    If $d_1=0$, then we define $L^{(l)}(d,l'_2,t+1) = L^{(l)}(d,l_2,t)$ for any $d\in[0,e_1]_{\integer}$. In the example of $t=18$ in Figure~\ref{fig:storage-tree-table-02}, we have $r'=19$ and $d_0=d_1=0$.
    Assume that $d_1\neq0$.
    If the left-cut of $C_{t+1}$ still remains in $\DD(dv,e_1,e_2,l_2,t)$, then we set $C''$ to be this left-cut. This case is depicted as the example of $t=15$ in Figure~\ref{fig:storage-tree-table-02} with $r'=16$, $d_0=4$, and $d_1=0$ and we obtain $C''=C_{13}$. In the case where $\DD(dv,e_1,e_2,l_2,t)$ contains no left-cut of $C_{t+1}$, we choose the latest marker $C'$ from $L^{(l)}(d_1,l_2,t)$  and define $l'=st\mbox{-}loc(C')$ and $t'=time(C')$. Since $d_1>d_0$, $C'$ is the left-representative of the completed $d_1$-section. See the example of $t=16$ in Figure~\ref{fig:storage-tree-table-02} with  $r'=17$, $d_0=0$, $d_1=2$, and $C'=C_8$. We then compute the left-cut $C''$ of $C_{t+1}$ that appears in a certain last-left-good $d'$-section with $d'<d_1$ satisfying both  $st\mbox{-}loc(C'')=l'_2$ and $t'< time(C'') <t$. This is done by recursively calling $\AAA$ as follows.
    Starting from $C'$ with $\DD(dv',d_1-1,e_2,l',t')$ with $dv'=k$, we inductively generate a pair of $\tilde{C}$ and $\DD(dv'-1,d_1-1,e_2,\tilde{l},\tilde{t})$ with $\tilde{t}=time(\tilde{C})$ and $\tilde{l}=st\mbox{-}loc(\tilde{C})$, and run $\AAA(\tilde{C},dv'-1,d_1-1,e_2, \DD(dv'-1,d_1-1,e_2,\tilde{l},\tilde{t}))$ to obtain the next marker $\tilde{C}'$ and its contingency tree $\DD(dv'-1,d_1-1,e_2,\tilde{l}',\tilde{t}+1)$ with $\tilde{l}'=st\mbox{-}loc(\tilde{C}')$ until we reach $C''$ with $dv' \geq dv(C'')$.
    Note that the final recursive call must be of the form $\AAA(\tilde{C},dv-1,d_1-1,e_2, \DD(dv-1,d_1-1,e_2,l'_2-1,\tilde{t}))$ with $l'_2-1=st\mbox{-}loc(\tilde{C})$. In the example of $t=16$ in Figure~\ref{fig:storage-tree-table-02}, $C''$ is $C_{10}$.

\item[(ii)] We inductively define $B_0=L^{(l)}_{cur}(0,l'_2,t')\cup\{C''\}$ and $B_{i+1} = L^{(l)}_{cur}(i,l'_2,t')\cup\{\bar{C}_i\}$ for any index $i\in[0,d_1-1]_{\integer}$, where $\bar{C}_i$ is the oldest marker in $B_i$. We define $L^{(l)}_{last}(i,l'_2,t+1) = B_i$ and $L^{(l)}_{cur}(i,l'_2,t+1) = \setempty$ for all  $i\in[0,d_1-1]_{\integer}$ and we then set $L^{(l)}_{cur}(d_1,l'_2,t+1) = \setempty$ and $L^{(l)}_{last}(d_1,l'_2,t+1) = L^{(l)}_{last}(d_1,l_2,t)\cup B_{d_1}$. In the example of $t=16$ in Figure~\ref{fig:storage-tree-table-02}, we obtain  $t'=10$ because of $C''=C_{10}$.

\item[(iii)] In the special case of $d_0>d_1$, we additionally define $E_0=\setempty$ and $E_{i+1} = L^{(l)}_{cur}(i,l_2,t)\cup\{\bar{C}'_i\}$ for any index $i\in[d_1+1,d_0-1]_{\integer}$, where $\bar{C}'_i$ is the oldest marker in $E_i$. We further set $L^{(l)}_{last}(i,l'_2,t+1) = E_i$,  $L^{(l)}_{cur}(i,l'_2,t+1) = \setempty$ for any $i\in[d_1+1,d_0-1]_{\integer}$, and  $L^{(l)}_{cur}(d_0,l'_2,t+1) = E_{d_0}$.
    We update all the others without changing the content of the old contingency lists $L^{(l)}(d,l_2,t)$. This process is exemplified in the case of $t=15$ in Figure~\ref{fig:storage-tree-table-02}, where $r'=16$, $d_0=4$, and $d_1=0$.

\item[(iv)] For all the other ``nodes'' in $\DD(dv,e_1,e_2,l_2,t)$, we update them to match the aforementioned changes.
\end{enumerate}

\item[(b)] [list $L^{(r)}(d',l'_2,t+1)$ for any  $d'\in[0,e_2]_{\integer}$] This process is obtained from Item  2(a) by replacing ``$l$'' with ``$r$'', ``$d$'' with ``$d'$'', and ``$e_1$'' with ``$e_2$''.

\item[(c)] [marker $C_{t+1}$] The desired $C_{t+1}$ is calculated in the following way.

\begin{enumerate}
  \setlength{\topsep}{-2mm}%
  \setlength{\itemsep}{1mm}%
  \setlength{\parskip}{0cm}%

\item[(i)] Assume that $dv\geq1$. If there is a right-cut $C''$ of $C_t$ at time $<t$, which is either stored inside $\DD(dv,e_1,e_2,l_2,t)$ or obtained by Item 3(a), then the desired marker $C_{t+1}$ is directly calculated from $C_t$ and $symb(C'')$ by applying $\delta$.
    Assume otherwise. If there is a left-representative $C_{rep}$ for $C_t$ at time $t_0<t$ and no left turn exists at location $\geq l'_2$ at time $<t_0$, then  we start with this left-representative $C_{rep}$ and compute a series of marker $\tilde{C}$ one by one (by incrementing time) from $C_{rep}$ using $\Box$ as inputs \emph{with no recursive call} until we obtain the $l'_2$-marker $C'$. We compute the desired $l'_2$-marker $C_{t+1}$ from $C_t$ and $symb(C')$  by applying $\delta$.

\item[(ii)] If $dv\leq 0$, then cell $l'_2$ must be already frozen blank $B$. In this case, since we do not need the left-cut of $C_t$, we set $\sigma'= B$ and computer $C_{t+1}$ from $C_t$ alone by applying $\delta$.
\end{enumerate}
\end{enumerate}
\end{enumerate}


\subsection{Complexity Analyses of the Procedures}\label{sec:complexity-analysis}

We have already given the description of Subroutine $\AAA$ in Section \ref{sec:description}.
What remains undone is to prove that the subroutine is correct and runs in polynomial time using only $O(\log^k{n})$ space.

In what follows, the \emph{depth of recursion} with respect to inputs of length $n$ refers to the maximum number of recursive calls invoked during any computation starting with any input of length $n$.

\begin{lemma}\label{correctness-proof}
Let $C_t$ and $C_{t+1}$ be two markers at times $t$ and $t+1$, respectively. Let $e_1,e_2\in[0,e_x]_{\integer}$ and let $n$ denote the input length.
If $dv\geq dv(C_{t+1})$, then $\AAA(C_t,dv,e_1,e_2,\DD(dv,e_1,e_2,l_2,t))$ correctly returns $(C_{t+1},dv,e_1,e_2,\DD(dv,e_1,e_2,l'_2,t+1))$ and  the depth of recursion is $O(\log{n})$.
\end{lemma}


\begin{proof}
We first argue the correctness of Subroutine $\AAA$. We assume that  $C_t = (q,l_1,l_2,\sigma,r,t)$ and $\DD(dv, e_1,e_2,l_2,t)$ have already been correctly computed. We run $\AAA$ on an input of the form $(C_t,dv,e_1,e_2,\DD(dv,e_1,e_2,l_2,t))$ with $dc\geq dv(C_{t+1})$.
During the computation of $\AAA$, Items 2(b), 2(c), 3(a), and 3(c) require series of recursive calls to $\AAA$ itself. Notice that those recursive calls are possible because a contingency tree in, e.g., Item 3(a) contains $\DD(dv-1,d_1-1,e_2,\tilde{l},\tilde{t})$ linked to a left-representative $C'$, which is included in $\DD(dv,e_1,e_2,l_2,t)$. Let $l'_2$ be defined from $l_2$ by Items 1--3.

\begin{claim}
For any $a\in\{l,r\}$, if $L^{(a)}(d,l_2,t)$ is correct, then so is $L^{(a)}(d,l'_2,t+1)$.
\end{claim}

\begin{proof}
Assume that $L^{(a)}(d,l_2,t)$ is correctly computed. It suffices to check the following five items: 1, 2(a), 2(b), 3(a), and 3(b). Since Item 1 is obvious, we then consider Item 2(a), in which the tape head moves to the right. Recall the parameters used in the description of Item 2(a).
If $r'$ is odd, then a $d$-section never ends at time $t$. Note that $d_1$ equals the previous value of $d_0$ at time $t-1$. Thus, for any $d\leq d_2$, we move all entries of $d$-section into $(d+1)$-section. By contrast, if $r'$ is even, then every $d$-section ends at time $t$ for any $d\leq d_0$.
Since the current $0$-section changes, we set it to empty at time $t+1$ and
move all entries in the current sections to their last-left-good sections by shifting the oldest markers in $d$-sections to $(d+1)$-sections.
Hence, $L^{(l)}(d,l'_2,t+1)$ is correctly computed.

Next, let us consider Item 3(a). If $d_1=0$, since the latest marker in $L^{(l)}(d_1,l_2,t)$ is the left-cut, it should be deleted.
Next, assume that $d_1>0$. If the left-cut remains after the deletion of the newest marker, then we simply shift all entries from the current sections to their last-left-good sections. Assume otherwise. After updating $L^{(l)}(d_1,l_2,t)$, the newest marker is the left-cut and is removed. A series of recursive calls helps us find the left-cut $C''$ of $C_{t+1}$.
Up to $d_1-1$, we move all entries of the current sections to their last-left-good sections by first adding the left-cut $C''$ and shifting the oldest markers in $d$-sections to $(d+1)$-sections.
In the case of $d_0>d_1$, for all $d\leq d_0$, any $d$-section begins at section time $t+1$. We move all entries in last-left-good sections to the current sections by shifting the oldest markers in $d$-sections to $(d+1)$-sections. Thus, $L^{(l)}(d,l'_2,t+1)$ contains the correct markers.

Items 2(b) and 3(b) are handled symmetrically to Items 3(a) and 2(a), respectively.
\end{proof}

\begin{claim}
If $L^{(a)}(d,l'_2,t+1)$ and $C_t$ are correctly computed for any $a\in\{l,r\}$ and any $i\in[0,t]_{\integer}$, then $C_{t+1}$ is also correctly computed.
\end{claim}

\begin{proof}
Assume that $L^{(a)}(d,l'_2,t+1)$ and $C_i$ are correctly computed for any $a\in\{l,r\}$ and any $i\in[0,t]_{\integer}$.
To verify the correctness of $C_{t+1}$, it suffices to check three items: 1, 2(c),  and 3(c). The case of Item 1 is obvious. Let us consider Item 2(c)(i). A left-cut $C''$ of $C_t$ at time $<t$ (if any) contains the information on the previously written storage symbol at cell $l'_2$. Thus, we can easily obtain $C_{t+1}$ correctly from this information together with $C_t$ by applying $\delta$. On the contrary, if there is no such left-cut, then we need only $C_t$ to compute $C_{t+1}$.
In the case of Item 3(c)(i), if there is a right-cut of $C_t$ at time $<t$, then we can compute $C_{t+1}$ correctly from this right-cut and $C_t$.  Otherwise, we can find an $l'_2$-marker $C'$ and then compute $C_{t+1}$ directly.
The remaining cases of Items 2(c)(ii) and 3(c)(ii) are obvious.
\end{proof}

Next, we argue that the depth of recursion of $\AAA$ is $O(\log{n})$ by examining all the processes of $\AAA$ that invoke recursive calls.
Let us first consider Item 3 of $\AAA$.
To compute $\DD(dv,e'_1,e'_2,l',t')$ in Item 3(a), if there is no right-cut $C''$ for $C_t$, then we need to make a series of recursive calls of the form $\AAA(\tilde{C},dv',d_1-1,e_2, \DD(dv',d_1-1,e_2,\tilde{l},\tilde{t}))$. Note that each recursive call depends on the parameter values of $dv'$ and $d_1$. Note that $d_1$ is smaller than $e_1$.
To compute $C_{t+1}$ with $l'_2=st\mbox{-}loc(C_{t+1})$ in Item 3(c), we need to compute $m$ markers $C^{(1)},C^{(2)},\ldots, C^{(m)}$ for a certain number $m\in\nat^{+}$ with $l'_2=st\mbox{-}loc(C^{(i)})$ and $0\leq time(C^{(i+1)})<time(C^{(i)})<t$ for any $i\in[m]$. If $m>dv(C_{t+1})$, then $\AAA$ eventually acknowledges that the cell $l'_2$ has already become frozen blank and it correctly computes $C_{t+1}$. At that point, the recursion ends. Thus, $m$ must be at most $k$.

A similar argument is applicable to Items 2(a) and 2(c).
Overall, since $e_1,e_2\in[0,e_x]_{\integer}$ and $e_x=\ceilings{\log|x|}$, the depth of recursion must be $O(\log{n})$.
\end{proof}

Finally, we prove Theorem \ref{upper-bound}, which provides an $\mathrm{SC}^k$-upper bound on the computational complexity of $k\mathrm{SDA}_{imm}$ for each index $k\geq2$.

\vs{-3}
\begin{proofof}{Theorem \ref{upper-bound}}
Since $\mathrm{SC}^k$ is closed under $\dl$-m-reductions, we immediately obtain $\mathrm{LOG}k\mathrm{SDA}_{imm}\subseteq \mathrm{SC}^k$ from $k\mathrm{SDA}_{imm}\subseteq \mathrm{SC}^k$. Hence, our goal here is to prove that $k\mathrm{SDA}_{imm}\subseteq \mathrm{SC}^k$.
Take an arbitrary language $L$ in $k\mathrm{SDA}_{imm}$ and any depth-immune $k$-sda $M$ that recognizes $L$ in polynomial time.
Consider the following procedure that simulates $M$ step by step on input $x$. Initially, we prepare the initial marker $C_0=(q_0,0,0,\triangleright,0)$ and the unique $k$-contingency list $\DD(k,e_x,e_x,0,0)$ linked to $C_0$.

\ms
\n{\sc Simulation Procedure $\PP$ for $M$:}
\begin{itemize}\vs{-2}
  \setlength{\topsep}{-2mm}%
  \setlength{\itemsep}{1mm}%
  \setlength{\parskip}{0cm}%

\item[] On input $x\in\Sigma^*$, set $n=|x|$ and $e_x=\ceilings{\log{|x|}}$, prepare $C_0$ and $\DD(k,e_x,e_x,0,0)$. Set $I_0=(C_0,k,e_x,e_x,\DD(k,e_x,e_x,0,0))$. Recursively, we run Subroutine $\AAA$ on $I_i$ and produce $I_{i+1}$ until $M$ halts at time, say, $t$ with $t \leq p(n)$. Let $I_t=(C_t,k,e_x,e_x,\DD(k,e_x,e_x,l_2,t))$ denote the final outcome of $\AAA$. If $C_t$ contains an accepting state, then we accept $x$; otherwise, we reject $x$.
\end{itemize}

By Lemma \ref{correctness-proof}, if $I_i$ correctly represents the configuration of $M$ at time $i$, then Subroutine $\AAA$ correctly computes $I_{i+1}$. By the mathematical induction, we conclude that $I_t$ correctly represents the halting configuration of $M$. This implies that the simulation procedure $\PP$ correctly determines whether or not $x$ belongs to $L$.

To store a contingency tree requires $O(\log^{k-1}{n})$ bits by Lemma \ref{marker-number} since each contingency list in the tree needs $O(\log{n})$ bits to express.
Lemma \ref{correctness-proof} then concludes that we need only polynomial runtime and $O(\log^k{n})$ memory bits to perform the simulation procedure $\PP$. Therefore, $L$ belongs to $\mathrm{SC}^k$.
\end{proofof}

\section{A Brief Discussion and Future Directions}

As access-controlled extensions of deterministic pushdown automata, we have introduced a new machine model of \emph{one-way deterministic depth-$k$ storage automata} (or $k$-sda's). We have then studied the fundamental properties of languages recognized by such machines and those of languages that are logarithmic many-one reducible (or $\dl$-m-reducible) to them.
We have called by  $k\mathrm{SDA}$ (resp., $k\mathrm{SDA}_{imm}$) the language family induced by depth-susceptible $k$-sda's (resp., depth-immune $k$-sda's).
This machine model naturally expands Hibbard's model of scan limited automata \cite{Hib67}. The $\dl$-m-closure of $k\mathrm{SDA}$ (resp., $k\mathrm{SDA}_{imm}$) is succinctly denoted by $\mathrm{LOG}k\mathrm{SDA}$ (resp., $\mathrm{LOG}k\mathrm{SDA}_{imm}$).
In Section \ref{sec:auxiliary}, we have presented two machine characterizations of $\mathrm{LOG}k\mathrm{SDA}$.
In Section \ref{sec:hardest-language}, we have also constructed a generic $\dl$-m-hard problem for $\mathrm{LOG}k\mathrm{SDA}$.
In Section \ref{sec:upper-bounds}, we have shown that $k\mathrm{SDA}_{imm} \subseteq \mathrm{SC}^k$ for each $k\geq2$.

As future research directions, we list six important questions that have been neither discussed nor solved in this exposition.

\begin{enumerate}
  \setlength{\topsep}{-2mm}%
  \setlength{\itemsep}{1mm}%
  \setlength{\parskip}{0cm}%

\item{} [class separations] As for Hibbard's $k$-lda's, it is known that, for every $k\in\nat^{+}$, $(k+1)$-lda's are in general more powerful in computational power than $k$-lda's \cite{Hib67}. Due to the similarity between $k$-lda's and $k$-sda's regarding the use of ``access-controlled'' memory devices, we can anticipate that, in general, $(k+1)$-sda's may have more computational power than $k$-sda's. It is thus  imperative to prove the difference between $k\mathrm{SDA}$ and $(k+1)\mathrm{SDA}$ for any index $k\geq2$.

\item{} [better upper bounds] We have shown in Section \ref{sec:upper-bounds} that $k\mathrm{SDA}_{imm}\subseteq \mathrm{SC}^k$ for any index $k\geq2$. However, we do not know whether $\mathrm{SC}^k$ is the best possible upper bound that we can achieve. We may ask whether or not the computational complexity of $k\mathrm{SDA}_{imm}$ is much lower than $\mathrm{SC}^k$. For instance, is it true that $k\mathrm{SDA}_{imm}\subseteq \mathrm{SC}^{k-1}$?

\item{} [nondeterminism and randomization] In this exposition, we have defined only the deterministic storage automata. Other important variants include a nondeterministic analogue of $k$-sda's, called \emph{depth-$k$ nondeterministic storage automata} (or $k$-sna's), and a randomized (or probabilistic) analogue of $k$-sda's, called \emph{depth-$k$ probabilistic storage automata} (or $k$-spa's). If we write $k\mathrm{SNA}$ and $k\mathrm{SBPA}$ to respectively denote the family of all languages recognized by $k$-sna's and that of all languages recognized by bounded-error $k$-spa's, then obviously $k\mathrm{SNA}$ is a subclass of $\np$ and $k\mathrm{SBPA}$ is a subclass of $\bpp$. It is important to study the basic properties of these language families and relationships between them.

\item{} [structural properties] Another important subject is (structural) closure properties under language operations (such as \emph{union}, \emph{intersection}, and \emph{concatenation}).  Since we have not discussed in this exposition any closure property of $k\mathrm{SDA}$ (as well as $k\mathrm{SDA}_{imm}$), we expect that a study on the closure properties may reveal fruitful features of $k\mathrm{SDA}$.

\item{} [depth-immune model] As noted in Section \ref{sec:LOGDCFL-beyond}, $\mathrm{LOGDCFL}$ is known to be characterized by numerous models, even with no use of $\dl$-m-reductions. Similarly, we have given in Section \ref{sec:auxiliary} two machine characterizations of  $\mathrm{LOG}k\mathrm{SDA}$. Is there any ``natural'' model that can  capture the depth-immune counterpart, $\mathrm{LOG}k\mathrm{SDA}_{imm}$, with no use of $\dl$-m-reductions?

\item{} [natural hard problems] In Section \ref{sec:hardest-language}, for each index $k\geq2$, we have constructed a language, $\mathrm{MEMB}_k$, which is hard for $\mathrm{LOG}k\mathrm{SDA}$ under $\dl$-m-reductions.  This  language $\mathrm{MEMB}_k$ is generic but looks quite artificial. It is therefore desirable to find ``natural'' hard (or even complete)  problems for $\mathrm{LOG}k\mathrm{SDA}$.
\end{enumerate}



\let\oldbibliography\thebibliography
\renewcommand{\thebibliography}[1]{%
  \oldbibliography{#1}%
  \setlength{\itemsep}{0pt}%
}
\bibliographystyle{plain}


\end{document}